\documentclass[letterpaper,journal]{IEEEtran}
\usepackage[T1]{fontenc}
\usepackage{amsmath, amssymb, amsfonts}
\usepackage{amsthm}
\usepackage{graphicx}
\usepackage{cite}
\usepackage{algorithm}
\usepackage{algpseudocode}
\DeclareMathOperator*{\argmax}{arg\,max}

\graphicspath{{./}{figs/}}
\providecommand{\safeincludegraphics}[2][]{%
  \IfFileExists{#2}{\includegraphics[#1]{#2}}{%
    \IfFileExists{figs/#2}{\includegraphics[#1]{figs/#2}}{\fbox{\texttt{Missing: #2}}}%
  }%
}

\usepackage{booktabs}
\usepackage{multirow}
\usepackage{array}
\usepackage{tabularx}
\usepackage{siunitx}
\usepackage{tikz}
\usepackage{pgfplots}
\pgfplotsset{compat=1.18}

\usepackage{dblfloatfix} 

\usepackage[hidelinks]{hyperref}
\hypersetup{
    pdftitle={Cooperative Edge Caching with Large Language Model in Wireless Networks},
}
\newtheorem{theorem}{Theorem}

\newtheorem{proposition}{Proposition}

\theoremstyle{definition}

\newcommand{\E}{\mathbb{E}}
\newcommand{\Ind}{\mathbf{1}} 
\newcommand{\NoOp}{\ensuremath{\text{\texttt{NoOp}}}}
\newcommand{\Invalid}{\ensuremath{\text{\texttt{Invalid}}}}

\title{Cooperative Edge Caching with Large Language Model in Wireless Networks}
\author{Ning~Yang,~\IEEEmembership{Member,~IEEE,} Wentao~Wang, Lingtao~Ouyang, and Haijun~Zhang,~\IEEEmembership{Fellow,~IEEE}%
\thanks{Corresponding author: Haijun Zhang.}%
\thanks{Ning Yang is with the Institute of Automation, Chinese Academy of Sciences, Beijing, 100190, China (e-mail: ning.yang@ia.ac.cn).}%
\thanks{Wentao Wang is with the Leicester International Institute, Dalian University of Technology, Dalian, Liaoning, 124221, China (e-mail: shiyanxi1@mail.dlut.edu.cn).}%
\thanks{Lingtao Ouyang is with the School of Management, Xi'an Jiaotong University, Xi'an, Shaanxi, 710049, China (e-mail: 2386558088@stu.xjtu.edu.cn).}%
\thanks{Haijun Zhang is with the School of Computer and Communication Engineering, University of Science and Technology Beijing, Beijing, 100083, China (e-mail: zhanghaijun@ustb.edu.cn).}%
}
\begin{document}
\maketitle

\begin{abstract}
Cooperative edge caching within spatially overlapping coverage zones introduces intricate coupling among Base Station (BS) caching decisions, rendering online content replacement highly sensitive to both spatial topology and temporal content reuse. Conventional heuristics often suffer from myopia, while Deep Reinforcement Learning agents typically rely on brittle numerical state/action representations and necessitate computationally prohibitive retraining under topological or traffic dynamics. This paper studies a centralized, cooperative multi-BS cache-replacement controller driven by a Large Language Model (LLM) within a deterministic text-to-action execution loop. At each time slot, the global cache state is deterministically rendered into a prompt encapsulating the cache inventory of each BS, the current set of deduplicated requests, and multi-scale frequency summaries. The LLM generates exactly one decision line per BS, after which a strict parser and feasibility checker either accept a feasible joint action or deterministically fall back to the all-BS no-operation (\NoOp) action. We align the LLM through a two-stage training paradigm: initially performing completion-only Supervised Fine-Tuning on look-ahead expert trajectories to acquire action syntax and robust initialization, followed by Group Relative Policy Optimization. The latter employs an ``opportunity-aware'' reward mechanism, utilizing the multi-step cooperative hit rate gain relative to a \NoOp\ baseline as the primary signal, alongside explicit penalties for invalid outputs. We focus on a reactive replacement setting with equal-sized files, at most one replacement per BS per slot, and insertion candidates restricted to current-slot requests. Under a frozen-trajectory evaluation protocol that reuses identical request traces and association graphs across methods, the proposed orchestrator approaches a single-step exhaustive-search reference (0.610 vs.\ 0.617 in a standard 5-BS scenario), surpasses classical baselines (e.g., +4.1\% over least-frequently used), and exhibits robust zero-shot transfer across cache capacity, library size, popularity skewness, and user density. Code is available at \url{https://github.com/gracefulning/CoopLLM-Cache}.
\end{abstract}


\section{Introduction}
\IEEEPARstart{M}{obile} data traffic in 5G and emerging 6G systems continues to increase rapidly due to the proliferation of video streaming, interactive applications, and content generated by artificial intelligence, placing sustained pressure on backhaul capacity and cloud infrastructures~\cite{CiscoVNI2019,Li2024SurveyEdgeCaching}. Edge caching alleviates this burden by storing frequently requested content closer to users and by updating cache inventories as demand evolves, thereby reducing latency and backhaul load~\cite{6763007,6600983}. While early research focused on isolated or non-overlapping cells, modern heterogeneous deployments are increasingly user-centric: coverage regions of multiple base stations (BSs) overlap, and users may be served by different BS subsets over time~\cite{8169053}. In such dense deployments, cache decisions across BSs interact non-trivially, rendering independent per-BS policies suboptimal.

Existing
edge caching strategies include popularity-based placement, probabilistic caching, and coded caching~\cite{6763007,Li2018HierarchicalEdgeCaching,Azimi2017FogEdge}. These methods typically rely on a stationary or slowly varying popularity model and often assume a fixed mapping between users and serving BSs. More recent designs incorporate spatial correlation and cooperative placement across BSs or Fog and Cloud layers~\cite{8169053,jiang2020deep}, but still depend on hand-crafted models (e.g., Zipf popularity, fixed mobility patterns) and task-specific analytical approximations. In realistic multi-cell deployments, however, popularity is heterogeneous across user groups, time-varying, and only observable through noisy request traces. This gap motivates learning-based control that can adapt to unknown popularity and network dynamics.
Learning-driven edge caching has therefore received increasing attention~\cite{Li2024SurveyEdgeCaching}. Supervised methods use demand prediction or collaborative filtering to infer future file popularity and then solve a deterministic placement problem~\cite{Zhang2023DeviceEdgeCachingIOTJ}. Reinforcement Learning (RL) approaches treat caching as a Markov Decision Process (MDP) and learn policies that maximize cache-hit related rewards using Q-learning, policy gradients, or actor--critic methods~\cite{Sadeghi2017OptimalScalableCaching,Chien2020QLearningEdgeCaching,Yang2020CacheAidedNOMA,Yang2025BeyondEdge}. Deep Reinforcement Learning (DRL) further integrates deep neural networks for value or policy approximation and has been widely explored for dynamic edge caching~\cite{Wu2019DynamicContentUpdate,Wang2024AdaptiveEdgeCaching,zhang2024communication,Wei2024MetaRLCoopCaching}. DRL-based designs can, in principle, cope with unknown and time-varying popularity and multi-cell coupling~\cite{8964499}. However, these systems remain tightly bound to a specific state and action parametrization and simulator, require large online interaction budgets for training, and often struggle to generalize when the network topology or traffic pattern changes drastically~\cite{9109574}.

On the other hand, Large Language Models (LLMs) have emerged as powerful general-purpose decision makers, leveraging their inherent reasoning capabilities to operate on textual descriptions of tasks, constraints, and feedback~\cite{yang2023foundationmodelsdecisionmaking}. Recent comprehensive surveys indicate that LLMs can support a wide range of telecom tasks, including configuration generation, fault diagnosis, and high-level optimization~\cite{Zhou2024LLMTelecomSurvey,Yang2025DecisionMakingSurvey}. In wireless networks, LLM agents have been proposed to orchestrate complex 6G functionalities~\cite{xu2024large}, drive semantic communication~\cite{Zhao2024LaMoSC}, and act as radio access network (RAN) applications (xApps) for resource management~\cite{Wu2025LLMxApp}, while LLM-enabled edge inference frameworks optimize the placement and scheduling of LLM computation across wireless edge servers~\cite{Zhang2025BeyondCloudLLM}. These works highlight the promise of LLMs as general decision engines that can ingest rich textual context, but they mostly use LLMs as high-level assistants, off-line solvers, or part of a hybrid pipeline, rather than as fully code-driven controllers interacting with the same environment interface used at evaluation time.

Motivated by the limitations of conventional heuristics in handling multi-cell interference and the poor generalization of DRL agents under drastic dynamics, we propose an LLM-based multi-BS edge caching framework in which a centralized coordinator uses the LLM as the learned decision module within a strict textual interface. The specific scope considered in this paper is centralized cooperative reactive cache replacement under three simplifying assumptions, following a constrained replacement setting similar to prior DRL-based edge-caching studies~\cite{Zhong2020DRLEdgeCaching}: equal file sizes, at most one replacement per BS per slot, and insertion candidates restricted to files requested in the current slot. To effectively align the general-purpose decision maker with this rigorous control task, we employ a two-stage training paradigm consisting of Supervised Fine-Tuning (SFT) on expert demonstrations followed by Group Relative Policy Optimization (GRPO)~\cite{Shao2024DeepSeekMath} for reward-driven improvement. Our main contributions are as follows:

\begin{itemize}
\item This paper formalizes a unified multi-BS cooperative edge caching environment characterized by overlapping cells, heterogeneous grouped Zipf user preferences, and multi-scale frequency statistics. Under the executable formulation studied here, the resulting online control problem is reactive cache replacement. Crucially, this formulation does not require any parametric popularity model. The environment exposes a strict textual state and action interface, ensuring that the LLM operates under the same constraints as practical controllers.
\item We develop a fully code-driven training pipeline utilizing a tailored GRPO fine-tuning scheme. We first initialize the policy via SFT to ensure the model produces valid syntax. Subsequently, we optimize an ``opportunity-aware'' reward that quantifies the multi-step cooperative hit rate gain relative to a \NoOp\ baseline. This design enables the LLM to learn complex coordination strategies while strictly adhering to validity constraints.
\item This study proposes a frozen-trajectory evaluation protocol to ensure fair comparison. For each traffic seed, the request process is fixed so that all methods process identical inputs under strict parsing and sanitization. This rigorous protocol eliminates randomness in request generation and allows for the precise quantification of both decision latency and policy quality.
\item Experimental results demonstrate that the proposed orchestrator effectively masters cooperative logic, reaching \textbf{98.9\%} of the single-step exhaustive-search reference in the standard 5-BS scenario and even slightly outperforming it in 2-BS settings (0.542 vs.\ 0.536). Moreover, the proposed approach exhibits robust zero-shot transferability across varying network scales and traffic statistics, highlighting its potential for flexible deployment in wireless networks without the need for frequent retraining.
\end{itemize}

\section{Related Work}

\subsection{Analytical and Heuristic Edge Caching}
The coded caching framework~\cite{6763007} established fundamental gains via structured placement, while systems like FemtoCaching~\cite{6600983} and device-to-device caching~\cite{Li2018HierarchicalEdgeCaching} showed that proximity-based storage substantially improves throughput and latency. Subsequent works extended these to fog-aided~\cite{Azimi2017FogEdge} and cooperative multi-cell architectures~\cite{8169053}. Although recent efforts incorporate temporal dynamics through freshness-aware caching~\cite{abolhassani2022fresh} and online learning~\cite{mhaisen2023online,wang2024similarity}, most designs rely on optimization under parametric demand assumptions, requiring relaxations for tractability. In cooperative multi-cell settings, strong spatial coupling and an exponential decision space often lead to decoupled approximations that sacrifice coordination gains.

To bridge caching decisions with wireless constraints, cross-layer optimization jointly considers communication, energy, and physical-layer mechanisms. Such formulations can optimize energy-latency trade-offs~\cite{abolhassani2022fresh,luo2022cooperative} and integrate advanced technologies like simultaneous transmitting and reflecting surfaces~\cite{hu2024joint}, but they typically yield complex mixed-integer programs. Consequently, these methods incur high computational overhead and require accurate demand models, limiting real-time deployment and adaptability to evolving constraints.

\subsection{Learning-Based and DRL-Based Edge Caching}
To cope with unknown and time-varying popularity, learning-based caching approaches leverage historical requests to infer future demand. Prediction-based methods first estimate future demand using time-series models or deep networks and then solve a deterministic placement problem~\cite{Zhang2023DeviceEdgeCachingIOTJ}. While effective when predictions are accurate, this two-stage pipeline can be brittle: forecast errors directly propagate to placement decisions, and retraining is often required when the topology or serving graph changes.

Reinforcement learning instead optimizes cache updates directly from interaction data. Sadeghi \textit{et al.}~\cite{Sadeghi2017OptimalScalableCaching} and Chien \textit{et al.}~\cite{Chien2020QLearningEdgeCaching} applied RL and Q-learning to small-cell and collaborative caching, while Wu \textit{et al.}~\cite{Wu2019DynamicContentUpdate} and Zhong \textit{et al.}~\cite{Zhong2020DRLEdgeCaching} leveraged deep function approximation to handle dynamic popularity and large discrete action spaces. To improve scalability in cooperative settings, Zhang \textit{et al.}~\cite{zhang2024communication} incorporated federated training, and multi-agent designs have been explored for distributed coordination~\cite{zhou2023twc_iov,zhou2023recommendation}. DRL has also been used for large-scale adaptive caching and coded caching under uncertain popularity~\cite{Wang2024AdaptiveEdgeCaching,zhang2020deep}. While DRL offers adaptability, it suffers from poor sample efficiency and overfitting to specific topologies. A minor change in the number of users or cache size often renders a trained DRL policy invalid, necessitating computationally expensive retraining from scratch~\cite{Cheng2023RLImpactSettings}, and the learned policies are often tied to specific numerical state/action encodings and simulators, limiting transferability without substantial retraining.

\subsection{LLM-Enabled Wireless Networking and Caching}
LLMs are increasingly explored as general-purpose assistants and controllers in telecommunications and wireless networks, supporting tasks such as configuration synthesis, code generation, and optimization-oriented reasoning~\cite{Zhou2024LLMTelecomSurvey,Yang2025DecisionMakingSurvey}. Representative examples include multimodal LLM agents for 6G networking~\cite{xu2024large}, LLM-enabled xApps for near-real-time RAN control~\cite{Wu2025LLMxApp}, and LLM-driven semantic communication systems for visual transmission~\cite{Zhao2024LaMoSC}. In parallel, edge inference frameworks such as ``Beyond the Cloud''~\cite{Zhang2025BeyondCloudLLM} study how to partition and schedule LLM inference across wireless edge nodes to meet latency and communication constraints. These studies highlight the promise of LLMs to reason over heterogeneous context, but also underscore practical challenges in grounding outputs to strict feasibility constraints and meeting tight runtime budgets. At the networking layer, Wu \textit{et al.}~\cite{Wu2024NetLLM} further demonstrated that LLM-centric workflows can be adapted to diverse network operation tasks, and broader discussions anticipate sustained impact on automation and troubleshooting pipelines~\cite{maatouk2024large}.

However, the intersection of LLMs and edge caching has thus far focused predominantly on \emph{model caching}, namely optimizing the partitioning and placement of LLM parameters or intermediate states at the edge to minimize inference latency~\cite{xu2024large,Zhang2025BeyondCloudLLM}. In these works, the LLM is the object to be cached rather than the decision engine for conventional content placement. Using LLMs as direct, online controllers for content caching, with outputs verifiably aligned to feasibility constraints under a strict execution interface, remains relatively underexplored.

\begin{figure}[htbp]
\centering
\includegraphics[width=\linewidth]{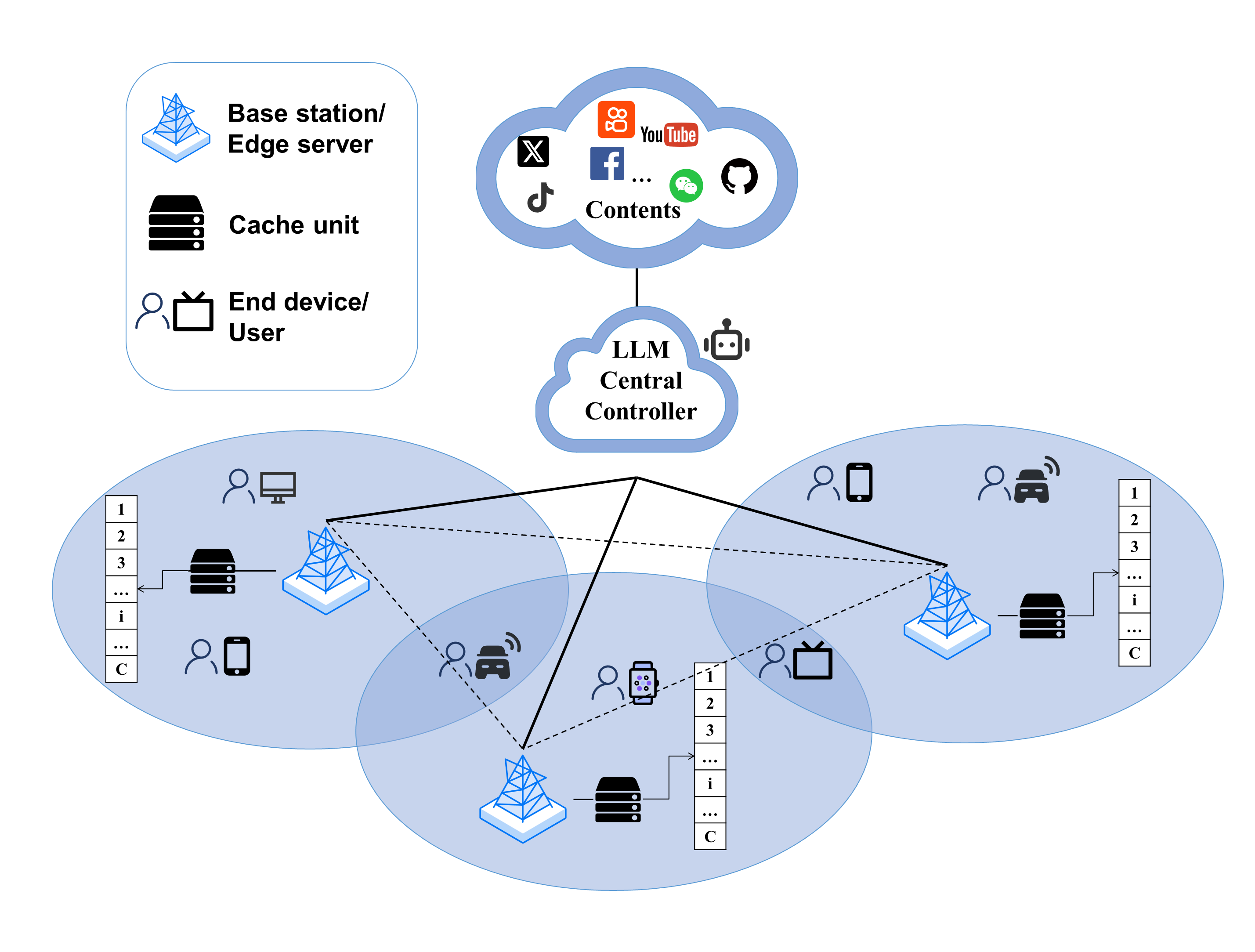}
\caption{Multi-cell cooperative edge caching architecture with overlapping BS coverage.}
\label{fig:overall}
\end{figure}

\section{Cooperative Multi-BS Cache Replacement: System Model and Problem Formulation}
\label{sec:system_model_problem}

\subsection{System Model}
We consider a cooperative multi-cell edge caching architecture as illustrated in Fig.~\ref{fig:overall}. The origin server stores the full content library $\mathcal{F}=\{1,\dots,F\}$ and serves requests that miss the cache through the backhaul. At the network edge, a set of $B$ cache-enabled Base Stations, denoted by $\mathcal{B}=\{1,\dots,B\}$, are deployed. Each BS $b \in \mathcal{B}$ is co-located with an edge server possessing a storage capacity of $C_b$ files, where all files are assumed to have equal size. Time is slotted with index $t\in\{1,2,\dots\}$. Although the caches are geographically distributed across BSs, the controller studied in this paper is centralized: at each slot, a single coordinator observes the global state and outputs one action for every BS.

Due to the dense deployment, coverage regions of BSs may overlap, meaning a user can be within the communication range of multiple BSs simultaneously. Let $\mathcal{U}^{(t)}$ denote the set of active users at slot $t$, and let $\mathcal{B}_u \subseteq \mathcal{B}$ denote the set of BSs covering user $u$. Under this cooperative framework, a user request is considered a cache hit and served locally if the content is available in any BS within the neighborhood $\mathcal{B}_u$; otherwise, it is retrieved from the origin server. We assume the user-BS association is fixed within each slot but may evolve slowly across slots due to user mobility.

\paragraph{Request Model}
Each user $u \in \mathcal{U}^{(t)}$ generates one content request $f_u^{(t)} \in \mathcal{F}$ per slot. While the underlying request process is unknown to the caching controller and may exhibit temporal concept drift, it is typically modeled using heavy-tailed distributions such as the Zipf distribution~\cite{6600983} to capture the realistic characteristic that a small subset of popular contents dominates the aggregate demand. For each BS $b$, the locally observed requests at slot $t$ constitute a multiset, as multiple users may request the same file concurrently. We capture this local request intensity by the request count $n_{b,f}^{(t)}$, defined as:
\begin{equation}
    \label{eq:request_count}
    n_{b,f}^{(t)} \triangleq \left|\left\{u \in \mathcal{U}^{(t)} :\, b \in \mathcal{B}_u,\, f_u^{(t)} = f\right\}\right|.
\end{equation}
To determine candidate contents for caching, we deduplicate this multiset into an admissible insertion set $\mathcal{R}^{(b)}_t$:
\begin{equation}
    \label{eq:admissible_set}
    \mathcal{R}^{(b)}_t \triangleq \{ f \in \mathcal{F} : n_{b,f}^{(t)} > 0 \} = \{ f_u^{(t)} : u \in \mathcal{U}^{(t)}, b \in \mathcal{B}_u \}.
\end{equation}
This set enforces a reactive caching policy where only contents requested within the current slot are eligible for insertion. Accordingly, the proposed scheme addresses online cache replacement rather than fully proactive prefetching. Meanwhile, the count $n_{b,f}^{(t)}$ serves as an auxiliary state feature that preserves the intensity of instantaneous demand.

\paragraph{Cache State and Update Model}
Let $x_{b,f}^{(t)} \in \{0,1\}$ be a binary indicator variable where $x_{b,f}^{(t)}=1$ implies that file $f$ is stored at BS $b$ during slot $t$. We denote the caching vector of BS $b$ as $\mathbf{x}_b^{(t)} = (x_{b,1}^{(t)}, \dots, x_{b,F}^{(t)})$. Consequently, the global caching configuration at slot $t$ is represented by the matrix $\mathbf{X}^{(t)} \in \{0,1\}^{B \times F}$, formed by stacking the vectors $\{\mathbf{x}_b^{(t)}\}_{b \in \mathcal{B}}$.

To ensure practical online operation and limit backhaul bandwidth consumption during content fetching, we adopt a constrained update interface. Specifically, each BS performs at most one cache update per slot. It can either maintain the current state, denoted as \NoOp, or perform a single-file replacement by evicting one existing file to insert one new file from the admissible set. Under this update model, the size of the local action space $D_b^{(t)}$ for BS $b$ is given by:
\begin{equation}
    D_b^{(t)} = C_b \cdot |\mathcal{R}^{(b)}_t| + 1.
\end{equation}
The action space consists of the \NoOp\ action plus all feasible distinct pairs of an eviction target and an insertion candidate.

\subsection{Problem Formulation}

Building on the system model, we formulate the cooperative caching control problem where a centralized coordinator updates the BS caches to maximize the long-term network-wide cache hit rate.

\textbf{Optimization Constraints}
The caching decisions must satisfy the following constraints at every slot $t$:
\begin{itemize}
    \item Storage Capacity: $\sum_{f=1}^{F} x_{b,f}^{(t)} \le C_b$ for all $b \in \mathcal{B}$.
    \item Update Limit: The transition from $\mathbf{X}^{(t-1)}$ to $\mathbf{X}^{(t)}$ must comply with the single-replacement rule. Mathematically, this restricts the Hamming distance between consecutive states such that $\sum_{f=1}^{F} |x_{b,f}^{(t)} - x_{b,f}^{(t-1)}| \le 2$ for all $b \in \mathcal{B}$.
\end{itemize}
Let $\mathbf{X} = \{\mathbf{X}^{(t)}\}_{t \ge 1}$ denote the sequence of caching decisions over time.

\textbf{Performance Metric: Cooperative Hit Rate}
Let $\mathcal{Q}^{(t)} = \{(u, f_u^{(t)}) : u \in \mathcal{U}^{(t)}\}$ denote the set of user-request pairs at slot $t$. A request is considered a cache hit if the requested file is available in the cache of at least one BS covering the user. Ideally, the network scheduler directs the user to a specific serving BS $b \in \mathcal{B}_u$ that holds the content. We define the cooperative availability indicator $\Ind(u, f_u^{(t)})$ as:
\begin{equation}
    \Ind(u, f_u^{(t)}) = 
    \begin{cases} 
        1, & \text{if } \sum_{b \in \mathcal{B}_u} x_{b,f_u^{(t)}}^{(t)} \ge 1, \\
        0, & \text{otherwise}.
    \end{cases}
\end{equation}
The instantaneous network cache hit rate, $P_{\mathrm{hit}}(t)$, is calculated as the ratio of satisfied requests to the total number of active users:
\begin{equation}
    \label{eq:phit_def}
    P_{\mathrm{hit}}(t) = \frac{1}{|\mathcal{U}^{(t)}|} \sum_{u \in \mathcal{U}^{(t)}} \Ind(u, f_u^{(t)}).
\end{equation}

\textbf{Optimization Objective}
We formulate the optimization problem $\mathbf{P}$ to identify the optimal decision sequence $\mathbf{X}$ that maximizes the long-term average cache hit rate. While real-world traffic patterns are dynamic, we define the objective under the standard stationarity and ergodicity assumption on $\{\mathcal{Q}^{(t)}\}$ so that the long-run time average exists almost surely:

\noindent 
\begin{subequations}
\label{prob:P1}
\begin{alignat}{2}
    \mathbf{P:} \quad & 
    \underset{\mathbf{X}}{\text{maximize}} \quad & & \lim_{T \to \infty} \frac{1}{T} \sum_{t=1}^{T} P_{\mathrm{hit}}(t) \label{eq:obj} \\
    & \text{subject to} & & \sum_{f=1}^{F} x_{b,f}^{(t)} \le C_b, \quad \forall b, t, \label{eq:cap} \\
    & & & \sum_{f=1}^{F} |x_{b,f}^{(t)} - x_{b,f}^{(t-1)}| \le 2, ~ \forall b, t \ge 2, \label{eq:single_swap} \\
    & & & x_{b,f}^{(t)} \in \{0, 1\}, \quad \forall b, f, t. \label{eq:bin}
\end{alignat}
\end{subequations}

Constraint \eqref{eq:cap} enforces storage limits, and \eqref{eq:bin} ensures binary variables. Constraint \eqref{eq:single_swap} formalizes the single-file replacement interface. Directly solving $\mathbf{P}$ is intractable; the joint caching problem over overlapping coverage is known to be nondeterministic polynomial-time hard even with known demand~\cite{6600983}, while practical deployment is further complicated by the lack of a priori knowledge of future requests and the presence of non-stationary distribution shifts~\cite{Sadeghi2017OptimalScalableCaching}. Consequently, the controller must employ a learning-based approach to adapt the policy $\mathbf{X}$ online based on observed request streams. In the following section, we cast this sequential decision problem as a MDP and develop an executable, LLM-based policy learning framework.

\begin{figure*}[htbp]
    \centering
    \includegraphics[width=0.95\textwidth]{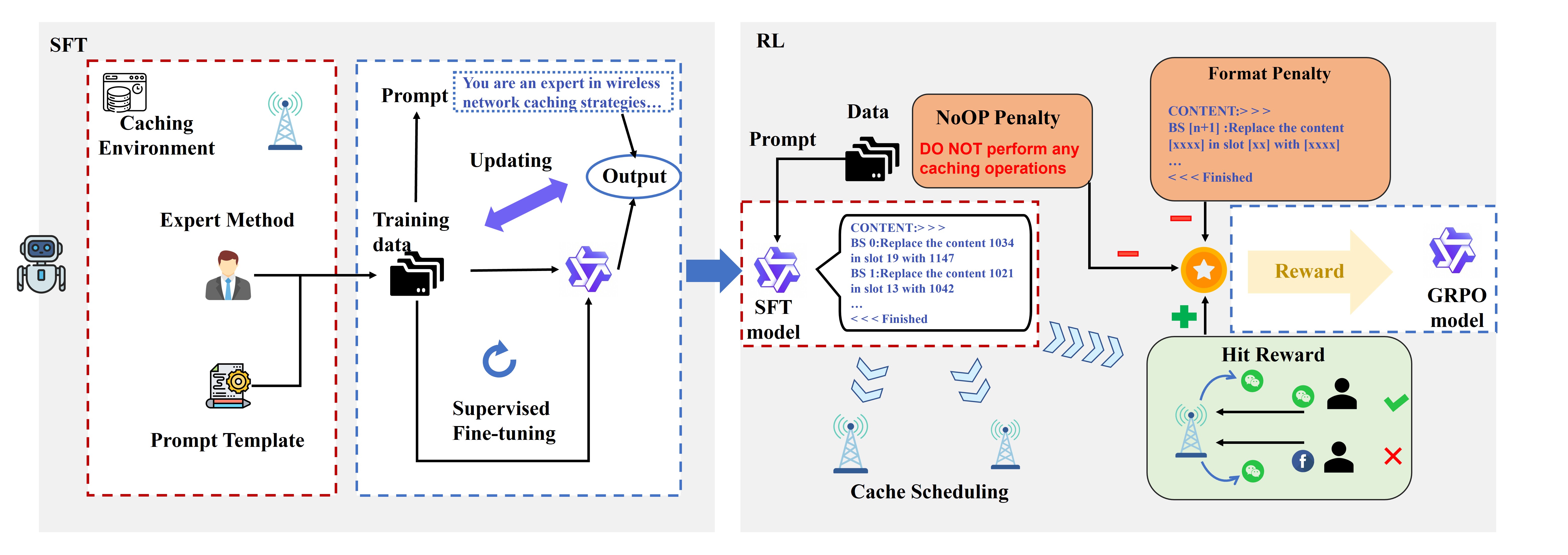} 
    \caption{Two-stage alignment pipeline. In Stage~I (SFT), the LLM is trained on expert trajectories to produce strictly executable cache-update actions. In Stage~II (GRPO), the policy is fine-tuned with a shaped multi-step gain reward together with explicit format and opportunity penalties, while actions are validated and executed through the same strict parser used at deployment.}
    \label{fig:training_pipeline}
\end{figure*}

\section{Methodology}
\label{sec:methodology}

This section details the proposed LLM-driven framework for cooperative multi-BS edge caching. As illustrated in Fig.~\ref{fig:training_pipeline}, we formulate the problem as a MDP with a strict textual interface. We implement a deterministic encode-parse-verify mechanism to instantiate the LLM as an executable centralized policy, aligned via a two-stage training paradigm: SFT on expert trajectories followed by GRPO.

\subsection{MDP Formulation and Executable Interface}
\label{subsubsec:related_definitions}

We model the online cache update process as a sequential decision process where a centralized coordinator observes the global network state and determines specific update actions for all BSs simultaneously.

\paragraph{State Space}
The state $s_t$ at slot $t$ encapsulates the instantaneous system status and historical demand context. Formally, we define the state as the tuple:
\begin{equation}
    s_t \triangleq \left( \mathbf{X}^{(t)}, \{\mathcal{R}^{(b)}_t\}_{b\in\mathcal{B}}, \boldsymbol{\phi}^{(t)} \right),
\end{equation}
where $\mathbf{X}^{(t)}$ represents the current global cache configuration, and $\mathcal{R}^{(b)}_t$ is the admissible insertion set for BS $b$ derived from current requests. To capture time-varying popularity trends, $\boldsymbol{\phi}^{(t)}$ denotes a collection of multi-scale request-frequency features. Let $\mathcal{W} \subset \mathbb{N}_{+}$ be a set of window lengths. For each file $f$ and BS $b$, we compute the historical appearance rate $\phi_{b,f}^{(t)}(w)$ over window $w \in \mathcal{W}$:
\begin{equation}
    \phi_{b,f}^{(t)}(w) \triangleq 
    \frac{1}{\min\{w,t\}}
    \sum_{\tau=\max\{1,t-w+1\}}^{t} \Ind(f \in \mathcal{R}^{(b)}_{\tau}\big),
\end{equation}
where $\Ind(\cdot)$ is the indicator function. By aggregating statistics over varying horizons, $\boldsymbol{\phi}^{(t)}$ enables the LLM to distinguish between transient noise and stable popularity shifts.

\paragraph{Action Space and Feasibility}
At each slot $t$, the decision for BS $b$, denoted by $a_{b,t}$, is either to maintain the current cache or to perform a single-file replacement. A replacement action is characterized by the tuple $(z_b, f_b^{\mathrm{in}}, f_b^{\mathrm{out}})$, where $z_b \in \{1,\dots,C_b\}$ is the cache slot index, $f_b^{\mathrm{in}}$ is the new content to insert, and $f_b^{\mathrm{out}}$ is the content to be evicted. The joint action is denoted by $a_t = (a_{1,t}, \dots, a_{B,t})$.

To ensure strict physical executability, we define a feasibility set function $\mathcal{A}(s_t)$. A joint action $a_t$ is valid, i.e., $a_t \in \mathcal{A}(s_t)$, if and only if every component $a_{b,t}$ satisfies the following constraints:
\begin{enumerate}
    \item \textbf{Admissibility:} The inserted content must be requested in the current slot:
    \begin{equation}
        f_b^{\mathrm{in}} \in \mathcal{R}^{(b)}_t.
    \end{equation}
    \item \textbf{Non-Duplication:} The content to be inserted must not already exist in the cache of BS $b$:
    \begin{equation}
        x_{b, f_b^{\mathrm{in}}}^{(t)} = 0.
    \end{equation}
    \item \textbf{Consistency:} The content targeted for eviction must explicitly match the file currently stored at the specified location. If slot $z_b$ stores file $f$, then:
    \begin{equation}
        f_b^{\mathrm{out}} = f \quad \text{s.t.} \quad x_{b,f}^{(t)} = 1.
    \end{equation}
\end{enumerate}

\paragraph{Optimization Objective}
Our ultimate objective remains the long-term average cooperative hit rate in Eq.~\eqref{eq:obj}. For policy learning, however, we adopt the following discounted surrogate objective~\cite{712192}:
\begin{equation}
\label{eq:objective_baseline}
    \max_{\pi} \;\; J(\pi) = \E_{\pi} \left[ \sum_{t=0}^{T-1} \gamma^{t} \, P_{\mathrm{hit}}(t) \right],
\end{equation}
where $\gamma \in (0,1]$ is the discount factor that balances immediate rewards against long-term cache utility. Under bounded rewards and ergodic continuing dynamics, choosing $\gamma$ close to one biases learning toward policies with strong long-run average performance, while also producing stable finite-horizon rollout targets for practical training. We therefore use Eq.~\eqref{eq:objective_baseline} only as a tractable training surrogate and report the undiscounted average cooperative hit rate in all experiments.

\paragraph{LLM-Executable Text-to-Action Interface}
To deploy an LLM parameterized by $\theta$ as the decision module of policy $\pi$, we formalize the control loop through three deterministic transformations:
\begin{itemize}
    \item \textbf{State-to-Prompt Encoder ($\mathcal{T}$):} The mapping $q_t = \mathcal{T}(s_t)$ converts the numerical state tuple---including cache content $\mathbf{X}^{(t)}$, request sets $\{\mathcal{R}^{(b)}_t\}$, and frequency features $\boldsymbol{\phi}^{(t)}$---into a structured textual prompt.
    \item \textbf{LLM Generation ($\pi_\theta$):} Conditioned on prompt $q_t$, the model generates a textual completion $o_t \sim \pi_\theta(\cdot|q_t)$.
    \item \textbf{Text-to-Action Parser ($\mathcal{P}$):} The parser $\mathcal{P}(o_t)$ first checks the required output format and then validates the decoded joint action against the feasibility set. If decoding succeeds and the joint action is feasible, it returns an executable action $a_t \in \mathcal{A}(s_t)$; otherwise, we denote the parser output as \Invalid\ and deterministically execute the all-BS \NoOp\ action. No heuristic fallback is used in the main protocol. During GRPO training, such invalid outputs additionally receive the formatting penalty defined in Eq.~\eqref{eq:penalty_term}.
\end{itemize}
Consequently, the induced executable policy is defined by the composition $s_t \xrightarrow{\mathcal{T}} q_t \xrightarrow{\pi_\theta} o_t \xrightarrow{\mathcal{P}} a_t$.

\subsection{Training Stage I: SFT}
\label{subsec:sft}

The primary objective of the first training stage is to bridge the gap between the open-ended generation capabilities of the LLM and the strict control interface of the caching environment. We define this code-alignment contribution as achieving two distinct goals:
\begin{enumerate}
    \item \textbf{Syntax Alignment:} Ensuring the model output $o_t$ strictly adheres to the format required by the parser $\mathcal{P}$, thereby minimizing syntax-based rejections.
    \item \textbf{Semantic Initialization:} Imparting the fundamental logic of validity and establishing a robust baseline policy derived from expert demonstrations.
\end{enumerate}

\subsubsection{Expert Data Generation}
We construct a completion-only dataset $\mathcal{D}_{\mathrm{sft}}=\{(q_t, y_t)\}$, where $q_t = \mathcal{T}(s_t)$ is the prompt rendering of state $s_t$, and $y_t$ is the target completion string. The target $y_t$ is obtained by serializing the expert joint action $a^*_t = (a^*_{1,t}, \dots, a^*_{B,t})$ into the requisite textual format.

To generate $a^*_t$, we employ a decoupled single-step exhaustive search expert. A full joint exhaustive search over all BSs is computationally intractable because the joint action space size, $|\mathcal{A}(s_t)| = \prod_{b=1}^{B} |\mathcal{A}_b(s_t)|$, grows exponentially with $B$ (as formally established in Proposition~\ref{prop:joint_action_growth}, Appendix~\ref{app:joint_exhaustive}), where $\mathcal{A}_b(s_t)$ denotes the feasible action set for BS $b$. Instead, the expert optimizes each BS $b$ independently by enumerating all actions in $\mathcal{A}_b(s_t)$ and selecting the one that maximizes the local contribution to the cooperative hit rate over a look-ahead horizon $H$. While this decoupled approach ignores inter-BS coupling, it serves to align the LLM with strict action validity and provides a robust initialization for the subsequent cooperative learning stage.

The data generation process, summarized in Algorithm~\ref{alg:sft_train}, begins with a warm-up phase of length $T_w$. This ensures that the cache states $\mathbf{X}^{(t)}$ and the historical frequency statistics $\boldsymbol{\phi}^{(t)}$ are fully populated and representative of steady-state operation before sampling begins.

\subsubsection{Supervised Fine-Tuning}
We fine-tune the model parameter $\theta$ using quantized low-rank adaptation (QLoRA)~\cite{dettmers2023qlora} with a rank of $r{=}64$. The training objective is to maximize the likelihood of generating the expert decision string $y_t$ given the state prompt $q_t$.
Let $y_t = (y_{t,1}, y_{t,2}, \dots, y_{t, L_t})$ be the token sequence of length $L_t$ corresponding to the target completion. The loss function $\mathcal{L}_{\mathrm{sft}}(\theta)$ is defined as the standard causal language modeling loss:
\begin{equation}
\label{eq:sft_loss}
\mathcal{L}_{\mathrm{sft}}(\theta)
=
-\E_{(q,y)\sim\mathcal{D}_{\mathrm{sft}}}
\left[
\sum_{n=1}^{L_t} \log \pi_{\theta}\!\left(y_n \mid q, y_{<n}\right)
\right],
\end{equation}
where $y_{<n}$ denotes the prefix of the target sequence.
To concentrate the model on the non-trivial replacement logic, we selectively collect SFT samples only when all BS caches are full, as this regime requires the agent to make complex eviction decisions rather than simple insertions. By contrast, cache-filling states contribute relatively few samples and are dominated by trivial insertions, which can dilute the replacement-oriented patterns that this stage is intended to teach.
Throughout SFT, we keep the state encoder $\mathcal{T}(\cdot)$ and the strict parser $\mathcal{P}(\cdot)$ fixed. This ensures that the training process directly optimizes the exact executable interface used during deployment.

\begin{algorithm}[t]
\caption{Stage I: SFT Data Generation and Training}
\label{alg:sft_train}
\begin{algorithmic}[1]
\Require Environment $\mathcal{E}$; encoder $\mathcal{T}$; horizon $H$; warmup $T_w$; target size $N$
\State Initialize dataset $\mathcal{D}_{\mathrm{sft}}\leftarrow \emptyset$
\State Reset $\mathcal{E}$ and run expert policy for $T_w$ steps (warm-up)
    \While{$|\mathcal{D}_{\mathrm{sft}}| < N$}
        \State Observe state $s_t$; compute request sets $\{\mathcal{R}^{(b)}_t\}$
    \For{$b=1$ to $B$}
        \State $a^*_{b,t} \leftarrow \textsc{LookaheadExpert}(\mathcal{E}, b, \mathcal{R}^{(b)}_t, H)$
    \EndFor
    \If{all BS caches are full}
        \State $q_t \leftarrow \mathcal{T}(s_t)$
        \State $y_t \leftarrow \text{Serialize}(\{a^*_{b,t}\}_{b=1}^{B})$ \Comment{Convert actions to text}
        \State $\mathcal{D}_{\mathrm{sft}} \leftarrow \mathcal{D}_{\mathrm{sft}} \cup \{(q_t, y_t)\}$
    \EndIf
    \State Step $\mathcal{E}$ using joint action $\{a^*_{b,t}\}_{b=1}^{B}$
\EndWhile
\State Train $\pi_{\theta}$ on $\mathcal{D}_{\mathrm{sft}}$ by minimizing Eq.~\eqref{eq:sft_loss}
\State \Return SFT model $\pi_{\theta_{\mathrm{sft}}}$
\end{algorithmic}
\end{algorithm}

\subsection{Training Stage II: GRPO Fine-Tuning}
\label{subsec:grpo}

While SFT establishes valid action syntax and initializes the policy, it relies on a decoupled approximation that ignores inter-BS correlations and lacks exploration. To transcend these limits, we employ GRPO to strictly optimize the cooperative objective. We select GRPO specifically because our reward structure measures fine-grained marginal gains over a \NoOp\ baseline; unlike value-based critics that struggle to estimate such subtle differences, GRPO's group normalization amplifies relative improvements, ensuring robust gradient signals without the complexity and instability of training a separate value network.

\subsubsection{Opportunity-Aware Reward Design}
The core of the RL stage is a reward mechanism that assigns credit to actions yielding deferred cooperative benefits while enforcing strict constraints. We construct a composite reward signal $R(o_t)$ based on the delta-weighted cooperative gain.

First, we define the look-ahead value function. Let $\mathbf{X}$ denote a generic global cache configuration. We generalize the instantaneous hit rate to $P_{\mathrm{hit}}(\mathbf{X}, \mathcal{Q})$, representing the hit rate of configuration $\mathbf{X}$ under request set $\mathcal{Q}$. Under the frozen-trajectory protocol, let $\mathcal{Q}^{(t+1:t+H)} \triangleq (\mathcal{Q}^{(t+1)}, \dots, \mathcal{Q}^{(t+H)})$ denote the fixed sequence of future requests. The normalized discounted look-ahead value $\mathcal{H}_H$ is defined as:
\begin{equation}
\label{eq:lookahead_val}
\mathcal{H}_H\big(\mathbf{X}; \mathcal{Q}^{(t+1:t+H)}\big) \triangleq \frac{\sum_{k=1}^{H} \gamma^{k-1} P_{\mathrm{hit}}\big(\mathbf{X}, \mathcal{Q}^{(t+k)}\big)}{\sum_{k=1}^{H} \gamma^{k-1}}.
\end{equation}
In our implementation, this look-ahead evaluator is realized by directly querying the frozen future request trace during the training phase, providing a computationally efficient and low-variance signal for policy optimization.

Next, we quantify the marginal improvement generated by the model's output. At slot $t$, let $o_t$ be the textual completion generated by the policy, and let $a_t \triangleq \mathcal{P}(o_t)$ denote the parsed joint action, where $a_t = \Invalid$ indicates a non-executable output. If $a_t$ is valid, the environment transitions from the current state $\mathbf{X}^{(t)}$ to a new state $\mathbf{X}^{(t+1)}$ according to the update rule: $\mathbf{X}^{(t+1)} = \text{Update}(\mathbf{X}^{(t)}, a_t)$. Otherwise, the executor deterministically falls back to the all-BS \NoOp\ action, so the cache state remains unchanged, i.e., $\mathbf{X}^{(t+1)}=\mathbf{X}^{(t)}$. The delta-weighted gain, $\Delta_{\text{perf}}(a_t)$, is the difference in future potential between the post-action state and the current state:
\begin{equation}
\label{eq:delta_perf_def}
\begin{aligned}
\Delta_{\text{perf}}(a_t) = & \mathcal{H}_H\big(\mathbf{X}^{(t+1)}; \mathcal{Q}^{(t+1:t+H)}\big) \\
& - \mathcal{H}_H\big(\mathbf{X}^{(t)}; \mathcal{Q}^{(t+1:t+H)}\big).
\end{aligned}
\end{equation}
This differential formulation acts as a baseline subtraction, ensuring that rewards reflect the quality of the decision rather than the intrinsic request intensity of the time slot.

To strictly enforce the interface and suppress suboptimal passivity, we define a penalty function $\Psi(o_t)$:
\begin{equation}
\label{eq:penalty_term}
\Psi(o_t) = 
\begin{cases} 
    \lambda_{\text{fmt}}, & \text{if } a_t = \Invalid \quad \text{(invalid output)}, \\
    \lambda_{\text{opp}}, & \text{if } a_t = \NoOp \text{ and } a^*_t \neq \NoOp, \\
    0, & \text{otherwise}.
\end{cases}
\end{equation}
where $\lambda_{\text{fmt}} < 0$ is a static formatting penalty, and $\lambda_{\text{opp}} < 0$ is a dynamic opportunity penalty. Here, $a^*_t$ is the expert joint action defined in Section~\ref{subsec:sft}; we use the condition $a^*_t \neq \NoOp$ as a training-time witness that a beneficial replacement exists, avoiding explicit counterfactual enumeration (see Section~\ref{subsec:reward_analysis}). This selectively discourages passivity without altering the preference ordering among valid active updates. For invalid outputs, the parser returns \Invalid, the executor applies the all-BS \NoOp\ fallback, and the resulting performance gain is therefore zero.

The final scalar reward $R(o_t)$ is the clipped sum of the gain and the penalty:
\begin{equation}
\label{eq:final_reward}
R(o_t) = \mathrm{clip}\big(\Delta_{\text{perf}}(a_t) + \Psi(o_t), -1, 1\big).
\end{equation}

\subsubsection{GRPO Objective}
We optimize the policy $\pi_{\theta}$ using GRPO, which eliminates the need for a separate value network critic. For each state prompt $q_t$, the behavior policy of the current update samples a group of $M$ independent completions $\{o_{t,i}\}_{i=1}^M$.
To reduce variance caused by non-stationary traffic patterns, we compute the group-relative Advantage $A_i$ by standardizing the rewards within the sampled group:
\begin{equation}
\label{eq:grpo_advantage}
A_i = \frac{R(o_{t,i}) - \mu_{\mathbf{R}}}{\sigma_{\mathbf{R}} + \epsilon},
\end{equation}
where $\mu_{\mathbf{R}}$ and $\sigma_{\mathbf{R}}$ are the mean and standard deviation of the reward set $\mathbf{R}=\{R(o_{t,j})\}_{j=1}^M$, and $\epsilon$ is a small stability constant.

The policy parameters $\theta$ are updated by maximizing the following surrogate objective, which incorporates a proximal policy optimization-style clipping mechanism~\cite{schulman2017proximal}. Let $\pi_{\text{old}}$ denote the behavior policy used to sample the group completions in the current GRPO update, and let $\pi_{\text{ref}}$ denote the fixed SFT initialization used only in the Kullback-Leibler (KL) regularizer:
\begin{equation}
\label{eq:grpo_obj}
\begin{aligned}
\mathcal{J}_{\text{GRPO}}(\theta) = \E_{q_t \sim \mathcal{D}} \Bigg[ & \frac{1}{M}\sum_{i=1}^{M} \min \Big( \rho_i(\theta)A_i, \text{clip}\big(\rho_i(\theta),\\ 1-\varepsilon, 1+\varepsilon\big)A_i \Big)
& - \beta D_{\mathrm{KL}}\big(\pi_{\theta}(\cdot|q_t) \| \pi_{\text{ref}}(\cdot|q_t)\big) \Bigg],
\end{aligned}
\end{equation}
where $\rho_i(\theta) = \frac{\pi_{\theta}(o_{t,i}|q_t)}{\pi_{\text{old}}(o_{t,i}|q_t)}$ is the unclipped likelihood ratio, $\varepsilon$ is the clipping range, and $\beta$ controls the strength of the KL regularization toward the fixed SFT reference.

\subsection{Reward Architecture Analysis}
\label{subsec:reward_analysis}
This subsection provides a potential-based reward shaping (PBRS)-style interpretation of the reward design. For clarity, the formal statements below analyze the executable reward surrogate with the outer clipping in Eq.~\eqref{eq:final_reward} set aside.

\subsubsection{Witness-Based Opportunity Penalty}
A key challenge is distinguishing intentional inaction (\NoOp) from missed improvements. Directly implementing the condition in Eq.~\eqref{eq:penalty_term} would require enumerating counterfactual write actions to verify if a gain-positive replacement exists, which is computationally prohibitive.

We avoid this by leveraging the supervised data pipeline: the GRPO stage utilizes training states where a reference action $a^*_t$ is available from the expert trajectory (Section~\ref{subsec:sft}). Let $a_t$ denote the parsed joint action from the model completion. We define the opportunity penalty $P_t(a_t)$ as:
\begin{equation}
    P_t(a_t) = 
    \begin{cases} 
    \lambda_{\text{opp}}, & \text{if } a_t = \NoOp \text{ and } a^*_t \neq \NoOp \\
    0, & \text{otherwise}
    \end{cases}
\end{equation}
where $\lambda_{\text{opp}} < 0$. Here, $a^*_t \neq \NoOp$ serves as a training-time \emph{witness} that a beneficial update exists, avoiding runtime enumeration of the action space. This targeted demotion only penalizes \NoOp\ when the reference expert identifies a gain-positive replacement, ensuring the relative ordering of valid write actions remains unaffected in the simplified analysis below.

\subsubsection{PBRS Interpretation and Witness-Based Opportunity Correction}
We now connect the multi-step gain term to PBRS-style intuition. Under the frozen-trajectory protocol, the gain in Eq.~\eqref{eq:delta_perf_def} can be viewed as a potential-difference shaping term~\cite{ng1999policy}. We therefore analyze the following executable reward surrogate:
\begin{equation}
\tilde{R}_t(a_t) = \Delta_{\text{perf}}(a_t) + P_t(a_t).
\end{equation}
Under the frozen-trajectory look-ahead protocol, we view the environment as having an augmented state $\bar{s}^{(t)}\triangleq (s_t,\mathcal{Q}^{(t+1:t+H)})$, i.e., appending the next $H$ request sets to the observable state.

\begin{theorem}[PBRS Invariance and Witness-Based Opportunity Correction]
\label{thm:shaping}
Define the state potential function as $\Phi(\bar{s}^{(t)}) \triangleq \mathcal{H}_H(\mathbf{X}^{(t)}; \mathcal{Q}^{(t+1:t+H)})$. For the executable reward surrogate $\tilde{R}_t$, the following properties hold:

\begin{enumerate}
    \item \textbf{Decision Invariance:} The performance gain $\Delta_{\text{perf}}$ constitutes a potential-difference signal. Maximizing the shaped reward is mathematically equivalent to maximizing the post-action potential:
 \begingroup
 \setlength{\abovedisplayskip}{3pt}
 \setlength{\belowdisplayskip}{3pt}
 \begin{equation}
\label{eq:invariance_derivation}
\begin{aligned}
    & \argmax_{a_t} \Delta_{\text{perf}}(a_t) \\
    = & \argmax_{a_t} \Big[ \Phi\big(\bar{s}^{(t+1)}(a_t)\big) - \Phi(\bar{s}^{(t)}) \Big] \\
    = & \argmax_{a_t} \Phi\big(\bar{s}^{(t+1)}(a_t)\big).
\end{aligned}
 \end{equation}
 \endgroup
    
    \item \textbf{Order Preservation and Witness-Based \NoOp\ Demotion:} The penalty $P_t$ preserves the relative ranking among write actions. Moreover, if the training-time witness satisfies $a_t^* \neq \NoOp$ and a write action $w$ yields positive gain, then:
    \begin{equation*}
    a_t^* \neq \NoOp \ \text{and}\ \Delta_{\text{perf}}(w) > 0 \implies \tilde{R}_t(w) > \tilde{R}_t(\NoOp).
    \end{equation*}
\end{enumerate}
\end{theorem}

The proof of Theorem \ref{thm:shaping} is provided in Appendix~\ref{app:proof_thm1}. These properties explain why the baseline-subtracted gain is a well-aligned surrogate for executable write decisions and why the witness-based penalty discourages degenerate passivity.

\section{Simulation Experiments and Result Analysis}
\label{sec:simulation}

\graphicspath{{figs/}}

\providecommand{\safeincludegraphics}[2][]{%
  \IfFileExists{#2}{\includegraphics[#1]{#2}}{%
    \IfFileExists{figs/#2}{\includegraphics[#1]{figs/#2}}{\fbox{\texttt{Missing: #2}}}%
  }%
}


We evaluate the proposed framework under a frozen-trajectory protocol for reproducibility and paired comparison. For each seed, we generate a task instance consisting of (i) a user--BS association graph induced by a geometric overlap model and (ii) a multi-user request trace, and then reuse the same instance across all methods. This eliminates performance variance caused by stochastic topology realizations and isolates decision quality under the same strict single-swap interface.

\subsection{Experimental Setup}
\label{subsec:exp_setup}

\subsubsection{Simulation Environment and Traffic Model}
We simulate the slotted cache-update process described in Section~\ref{sec:system_model_problem}. We consider two network scales: a two-BS setting for ablations and a five-BS setting for scalable coordination. In each instance, BS locations are fixed on a simple layout, while user locations are sampled uniformly over the service area; each user is connected to all BSs within a coverage radius, yielding overlapping neighborhoods and non-trivial cooperative coupling.

User requests follow a grouped-Zipf model: each group draws a distinct, randomly permuted Zipf ranking over the library and users inherit their group's distribution. This construction creates heterogeneous and spatially non-uniform popularity, different BSs observe different mixtures of user groups, while retaining the heavy-tailed demand patterns commonly assumed in edge caching.

To avoid cold-start artifacts and ensure stable historical statistics, we warm-start each frozen instance by running a fixed expert prefill policy for $100$ slots to populate caches and history buffers. All controllers are then evaluated from the same warm-started state under the same frozen requests and association graph. Accordingly, the reported results characterize steady-state replacement performance rather than cold-start or partially filled-cache behavior.

\subsubsection{Implementation Details}
We instantiate the controller with \textbf{Qwen2.5-7B-Instruct}~\cite{qwen2025qwen25technicalreport} and fine-tune it via QLoRA to keep the backbone fixed while adapting to the strict prompt-action interface. For evaluation, we use greedy decoding together with the same parser as in training to remove generation randomness and ensure strict executability. If a generated completion fails parsing or feasibility validation, the executed joint action deterministically defaults to the all-BS \NoOp\ action.
Default parameters are summarized in Table~\ref{tab:sim_params}; unless explicitly swept, we use $\mathcal{W}=\{10,100,1000\}$ to expose short, medium, long-term demand signals.

\begin{table}[t]
  \centering
  \caption{Default simulation parameters (unless explicitly swept).}
  \label{tab:sim_params}
  \small
  \renewcommand{\arraystretch}{1.12}
  \begin{tabularx}{\columnwidth}{@{}X l@{}}
    \toprule
    Parameter & Value \\
    \midrule
    BSs $B$ & $2$, $5$ \\
    Users $U$ & $20$ (two-BS), $40$ (five-BS) \\
    Library size $F$ & $100$ (unless swept) \\
    Cache size $C_b$ & $10$ (all $b$, unless swept) \\
    Groups $G$ & $3$ \\
    Zipf skew $\alpha$ & $1.2$ (unless swept) \\
    History windows $\mathcal{W}$ & $\{10,100,1000\}$ slots \\
    Evaluation decoding & greedy ($T{=}0$) \\
    \bottomrule
  \end{tabularx}
\end{table}

\subsubsection{Baselines}
We benchmark the proposed LLM controllers against representative reactive eviction heuristics, a myopic exhaustive reference, and multiple DRL baselines. All methods operate under the same admissible insertion sets $\mathcal{R}^{(b)}_t$ and the same single-swap update constraint; the LLM-based controller additionally uses strict parsing and feasibility verification to guarantee executability. To broaden the RL comparison under the same centralized executable interface, we include Soft Actor-Critic (SAC), Actor-Critic (AC), Deep Deterministic Policy Gradient (DDPG), and a distributional reinforcement learning (Distributional RL) baseline. Although many multi-agent cooperative caching methods exist, many assume decentralized execution or centralized training with decentralized execution, with observation and execution models that differ from ours; we therefore focus on baselines that can be instantiated under the same frozen-trajectory protocol and strict single-swap action interface.
\begin{itemize}
    \item \textbf{Least Recently Used (LRU)}~\cite{ahmed2013analyzing}: This policy evicts the least recently requested cached content.
    \item \textbf{Least Frequently Used (LFU)}~\cite{jaleel2010high}: This strategy evicts the least frequently requested content based on accumulated counts, although it remains vulnerable to cache pollution under non-stationary demand.
    \item \textbf{First-In First-Out (FIFO)}~\cite{rossi2011caching}: This method evicts the oldest cached content based strictly on insertion order.
    \item \textbf{SAC}~\cite{haarnoja2018soft}: We employ the soft actor-critic baseline as a maximum-entropy DRL method, trained with a numerical state/action encoding under the same constrained action interface.
    \item \textbf{Actor-Critic (AC)}~\cite{konda1999actor}: This baseline jointly learns a policy network and a value critic under the same centralized observation and constrained action interface.
    \item \textbf{Deep Deterministic Policy Gradient (DDPG)}~\cite{lillicrap2020continuous}: This baseline uses an actor--critic architecture with deterministic policy updates to test whether continuous-control style training improves cache-update quality.
    \item \textbf{Distributional RL}~\cite{bellemare2017distributional}: This baseline uses a distributional reinforcement learning formulation under the same frozen trajectories and strict execution constraints.
    \item \textbf{Single-step Exhaustive Search}: This myopic reference enumerates \NoOp\ and all feasible single-slot replacements for each BS to maximize the immediate cooperative hit rate. While its per-slot cost scales as $\mathcal{O}(C_b|\mathcal{R}^{(b)}_t|)$, extending it to full joint enumeration across BSs is infeasible due to the exponential growth of the joint action space (see Appendix~\ref{app:joint_exhaustive}).
\end{itemize}

\begin{table*}[!t]
  \centering
  \caption{Two-BS performance on three frozen seeds. Prefix averages at slots $\{50,100,\dots,300\}$ and overall mean. Methods are grouped into LLM-based controllers, DRL baselines, and classical/myopic references. The \textbf{bold} and \underline{underlined} values represent the best and second-best performance, respectively.}
  \label{tab:two_bs_stepwise}
  \small
  \sisetup{detect-weight=true,detect-family=true,table-number-alignment=center}
  \setlength{\tabcolsep}{3pt}
  \renewcommand{\arraystretch}{1.18}
  \begin{tabular*}{\textwidth}{@{\extracolsep{\fill}} l 
    S[table-format=1.3] S[table-format=1.3] S[table-format=1.3] 
    S[table-format=1.3] S[table-format=1.3] S[table-format=1.3] S[table-format=1.3]}
    \toprule
    \multirow{2}{*}{Method} & \multicolumn{7}{c}{Average Cooperative Hit Rate ($\uparrow$)} \\
    \cmidrule(lr){2-8}
    & \multicolumn{1}{c}{Slot 50} & \multicolumn{1}{c}{Slot 100} & \multicolumn{1}{c}{Slot 150} 
    & \multicolumn{1}{c}{Slot 200} & \multicolumn{1}{c}{Slot 250} & \multicolumn{1}{c}{Slot 300} 
    & \multicolumn{1}{c}{Mean} \\
    \midrule
    GRPO LLM       & \bfseries 0.530 & \bfseries 0.536 & \bfseries 0.538 & \bfseries 0.545 & \bfseries 0.550 & \bfseries 0.552 & \bfseries 0.542 \\
    SFT LLM        & 0.522 & 0.526 & 0.525 & 0.534 & 0.538 & 0.540 & 0.531 \\
    \specialrule{0.08em}{0.10em}{0.10em}
    AC             & 0.448 & 0.449 & 0.454 & 0.455 & 0.467 & 0.486 & 0.460 \\
    DDPG           & 0.504 & 0.505 & 0.505 & 0.505 & 0.510 & 0.517 & 0.508 \\
    SAC            & 0.462 & 0.474 & 0.477 & 0.481 & 0.482 & 0.482 & 0.477 \\
    Distributional RL & 0.503 & 0.503 & 0.504 & 0.503 & 0.508 & 0.512 & 0.506 \\
    \specialrule{0.08em}{0.10em}{0.10em}
    Exhaustive     & {\underline{0.526}} & {\underline{0.531}} & {\underline{0.536}} & {\underline{0.540}} & {\underline{0.543}} & {\underline{0.541}} & {\underline{0.536}} \\
    LFU            & 0.471 & 0.494 & 0.498 & 0.503 & 0.506 & 0.507 & 0.497 \\
    LRU            & 0.491 & 0.499 & 0.503 & 0.507 & 0.508 & 0.507 & 0.502 \\
    FIFO           & 0.318 & 0.312 & 0.308 & 0.317 & 0.314 & 0.309 & 0.313 \\
    \bottomrule
  \end{tabular*}
\end{table*}

\subsection{Reward Design Validation}
\label{subsec:reward_ablation}

\begin{figure}[t]
  \centering
  \safeincludegraphics[width=0.9\columnwidth]{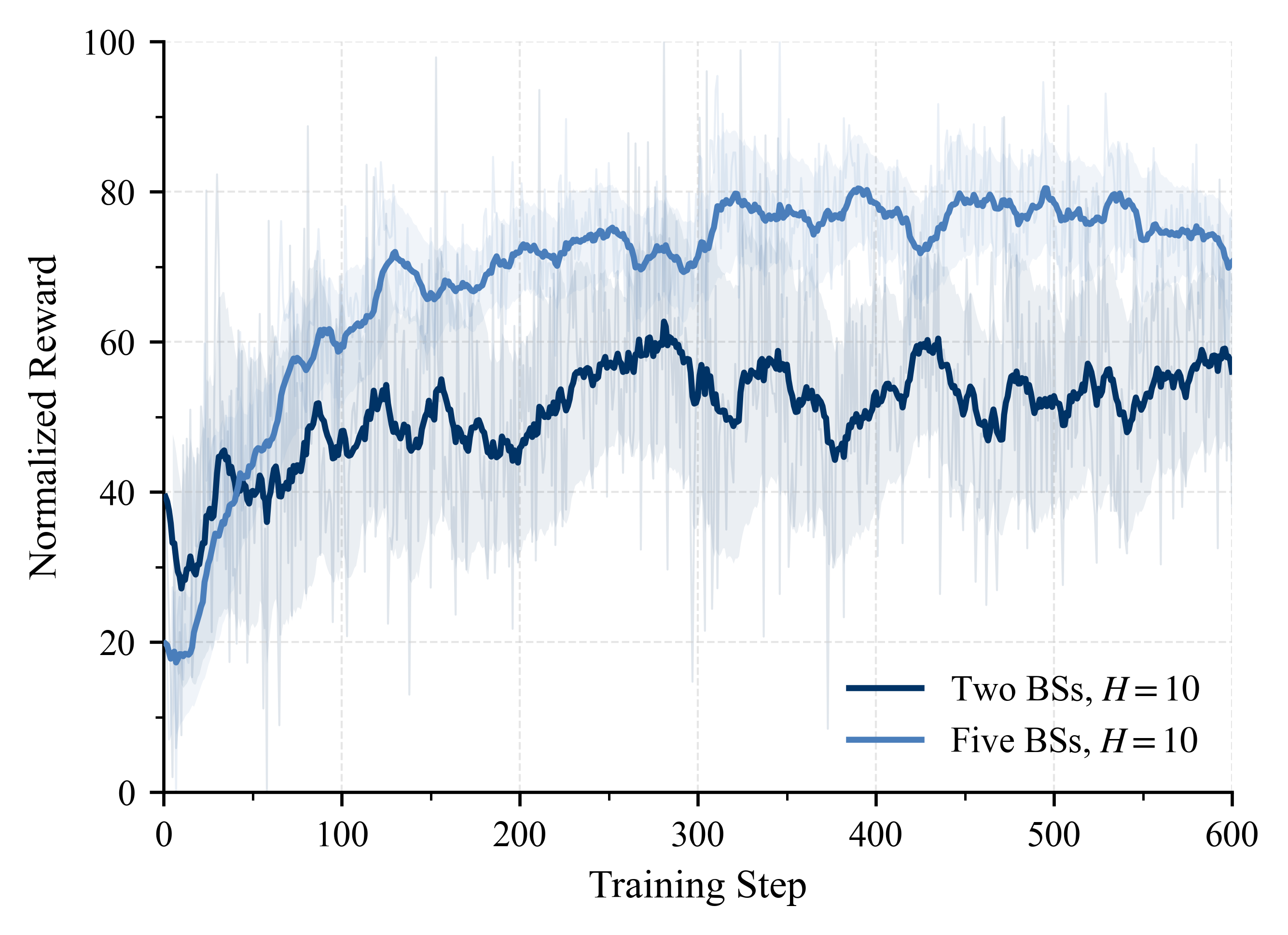}
  \caption{Representative GRPO training trajectories under strict executability (normalized reward vs.\ training step) for two-BS and five-BS settings.}
  \label{fig:reward_convergence_clean}
\end{figure}

\begin{figure}[t]
  \centering
  \safeincludegraphics[width=0.9\columnwidth]{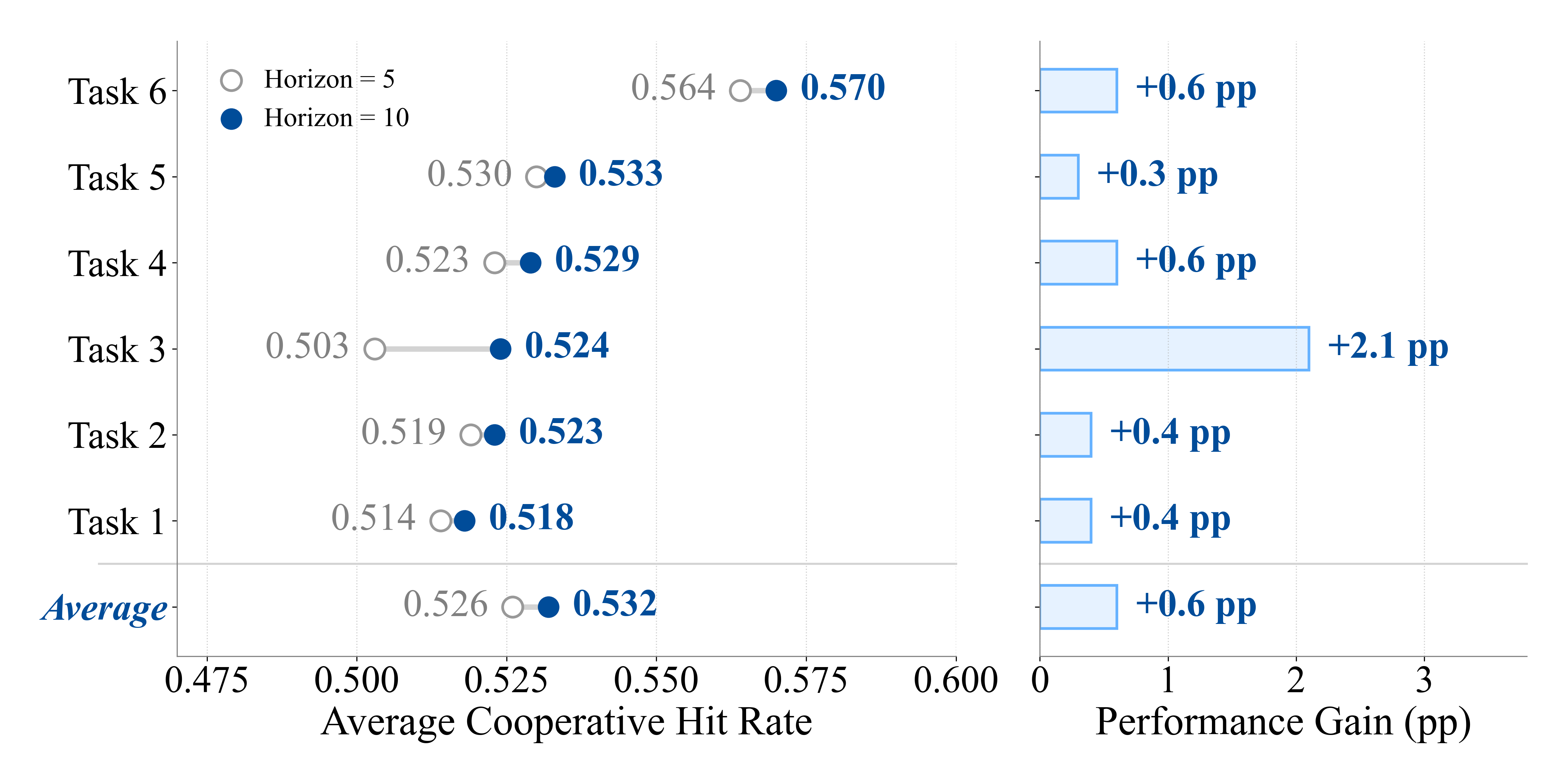}
  \caption{Two-BS horizon comparison across three frozen tasks: training with $H{=}10$ consistently improves cooperative hit rate over $H{=}5$.}
  \label{fig:two_bs_reward_horizon_singlecol}
\end{figure}

Before conducting the full benchmark, we validate the GRPO reward design to ensure stable optimization.
Directly optimizing the raw cooperative hit rate can induce high-variance updates, as the absolute metric fluctuates substantially with request diversity and overlap patterns, obscuring the true quality of decisions.
We therefore adopt a baseline-subtracted reward that measures the relative improvement over a \NoOp\ baseline under the same cache and request state.
As illustrated in Fig.~\ref{fig:reward_convergence_clean}, this normalized signal enables the training curves to rise steadily and plateau without oscillation in both two-BS and five-BS settings, confirming that the constrained text-to-action policy is reliably learnable.
Furthermore, we investigate the impact of the look-ahead horizon $H$. As shown in Fig.~\ref{fig:two_bs_reward_horizon_singlecol}, extending the horizon from $H{=}5$ to $H{=}10$ consistently improves the cooperative hit rate. A longer horizon captures delayed benefits, such as reduced redundancy across overlapping BSs, that are invisible to shorter-term feedback.
Consequently, all subsequent experiments utilize the baseline-subtracted reward with $H{=}10$.

\begin{table*}[!t]
  \centering
  \caption{Five-BS performance on three frozen seeds. Prefix averages at slots $\{50,100,\dots,300\}$ and overall mean. Methods are grouped into LLM-based controllers, DRL baselines, and classical/myopic references. The \textbf{bold} and \underline{underlined} values represent the best and second-best performance, respectively.}
  \label{tab:five_bs_stepwise}
  \small
  \sisetup{detect-weight=true,detect-family=true,table-number-alignment=center}
  \setlength{\tabcolsep}{3pt}
  \renewcommand{\arraystretch}{1.18}
  \begin{tabular*}{\textwidth}{@{\extracolsep{\fill}} l 
    S[table-format=1.3] S[table-format=1.3] S[table-format=1.3] 
    S[table-format=1.3] S[table-format=1.3] S[table-format=1.3] S[table-format=1.3]}
    \toprule
    \multirow{2}{*}{Method} & \multicolumn{7}{c}{Average Cooperative Hit Rate ($\uparrow$)} \\
    \cmidrule(lr){2-8}
    & \multicolumn{1}{c}{Slot 50} & \multicolumn{1}{c}{Slot 100} & \multicolumn{1}{c}{Slot 150} 
    & \multicolumn{1}{c}{Slot 200} & \multicolumn{1}{c}{Slot 250} & \multicolumn{1}{c}{Slot 300} 
    & \multicolumn{1}{c}{Mean} \\
    \midrule
    GRPO LLM       & {\underline{0.615}} & {\underline{0.615}} & {\underline{0.609}} & {\underline{0.607}} & {\underline{0.606}} & {\underline{0.607}} & {\underline{0.610}} \\
    SFT LLM        & 0.595 & 0.595 & 0.587 & 0.588 & 0.584 & 0.584 & 0.589 \\
    \specialrule{0.08em}{0.10em}{0.10em}
    AC             & 0.566 & 0.565 & 0.567 & 0.564 & 0.569 & 0.573 & 0.567 \\
    DDPG           & 0.557 & 0.555 & 0.559 & 0.556 & 0.560 & 0.566 & 0.559 \\
    SAC            & 0.561 & 0.554 & 0.552 & 0.551 & 0.550 & 0.556 & 0.554 \\
    Distributional RL & 0.515 & 0.512 & 0.513 & 0.513 & 0.532 & 0.553 & 0.523 \\
    \specialrule{0.08em}{0.10em}{0.10em}
    Exhaustive     & \bfseries 0.626 & \bfseries 0.617 & \bfseries 0.616 & \bfseries 0.614 & \bfseries 0.616 & \bfseries 0.612 & \bfseries 0.617 \\
    LFU            & 0.597 & 0.591 & 0.586 & 0.584 & 0.579 & 0.580 & 0.586 \\
    LRU            & 0.582 & 0.579 & 0.577 & 0.571 & 0.567 & 0.566 & 0.574 \\
    FIFO           & 0.363 & 0.369 & 0.378 & 0.375 & 0.368 & 0.366 & 0.370 \\
    \bottomrule
  \end{tabular*}
\end{table*}

\subsection{Performance Evaluation}
\label{subsec:performance}

We report results on three frozen instances (seeds) per scenario, where both the request trace and the user--BS association graph are fixed as described in Section~\ref{subsec:exp_setup}. After warm-starting, each controller is rolled out for $T{=}300$ decision slots under the strict single-swap constraint. To summarize both early adaptation and steady-state stability within this warm-started setting, we report the prefix-average cooperative hit rate,
$\overline{P}_{\mathrm{hit}}(t)\triangleq \frac{1}{t}\sum_{\tau=1}^{t} P_{\mathrm{hit}}(\tau)$,
at $t\in\{50,100,\dots,300\}$ and the overall mean.

\paragraph{Two-BS Scenario}
Table~\ref{tab:two_bs_stepwise} confirms that GRPO consistently outperforms SFT and all baselines throughout the rollout. Remarkably, GRPO slightly exceeds even the single-step exhaustive reference. This advantage stems from a fundamental objective mismatch: while the exhaustive policy myopically maximizes the immediate hit rate, GRPO is trained on a multi-step gain. This training signal enables the agent to make strategic replacements that, while locally suboptimal in the current slot, effectively reduce future content redundancy across overlapping coverage zones.
In contrast, heuristic rules like LRU and LFU perform poorly because they rely solely on local recency or frequency signals. Lacking mechanisms to reason about network-level complementarity, they often replicate popular content across neighboring BSs, wasting aggregate storage capacity.
Among the DRL baselines, even the strongest one, DDPG, reaches only 0.508 on average, followed closely by Distributional RL at 0.506, while SAC and AC remain at 0.477 and 0.460, respectively. This consistent gap indicates that purely numerical RL baselines still struggle to learn the cooperative replacement logic captured by the LLM controllers. Without the structural priors provided by SFT and the stable, normalized gain signal used by GRPO, these RL baselines must explore a large, discrete, and feasibility-constrained action space using comparatively high-variance rewards, which limits their final performance.

\paragraph{Five-BS Scenario}
In the five-BS scenario, the coordination complexity intensifies significantly due to intricate spatial coupling and the exponential expansion of the joint action space.
Despite these challenges, the GRPO-tuned framework demonstrates robust scalability, achieving a mean cooperative hit rate of \textbf{0.610}.
This performance approaches the single-step exhaustive reference, reaching \textbf{98.9\%} of its mean hit rate (0.610 vs.\ 0.617).
Crucially, while this reference relies on computationally intensive online enumeration to maximize immediate returns, our LLM-based controller infers high-quality decisions directly from the prompt without the overhead of iterative search.
The method establishes a substantial margin over both the SFT initialization (0.589) and the best heuristic, LFU (0.586).
Among the DRL baselines, AC is the strongest at 0.567, followed by DDPG at 0.559, SAC at 0.554, and Distributional RL at 0.523.
This result reinforces the same conclusion observed in the two-BS case: once the action space and spatial coupling become more complex, standard RL baselines struggle to scale under absolute-reward training, whereas our framework benefits from SFT priors and the variance-reduced gain-based optimization of GRPO.

\begin{table}[t]
  \centering
  \caption{Two-BS zero-shot generalization across cache capacities. Best per row is bolded.}
  \label{tab:two_bs_zero_shot_cache}
  \small
  \setlength{\tabcolsep}{4.8pt}
  \renewcommand{\arraystretch}{1.14}
  \begin{tabular}{c r r r r r r}
    \toprule
    \multirow{2}{*}{Cache $C_b$} & \multicolumn{3}{c}{Heuristics} & \multicolumn{3}{c}{LLM / Reference} \\
    \cmidrule(lr){2-4}\cmidrule(lr){5-7}
    & FIFO & LRU & LFU & Exhaustive & SFT & GRPO \\
    \midrule
    10 & 0.289 & 0.488 & 0.501 & 0.521 & 0.531 & \textbf{0.544} \\
    15 & 0.371 & 0.589 & 0.598 & 0.616 & 0.612 & \textbf{0.624} \\
    20 & 0.440 & 0.669 & 0.674 & 0.681 & 0.675 & \textbf{0.683} \\
    25 & 0.501 & 0.729 & 0.728 & \textbf{0.739} & 0.731 & 0.734 \\
    30 & 0.555 & 0.771 & 0.771 & \textbf{0.775} & 0.764 & 0.770 \\
    \bottomrule
  \end{tabular}
\end{table}

\subsection{Generalization and Robustness Tests}
\label{subsec:generalization}

We assess the zero-shot generalization of a single trained controller under varying cache capacities, library sizes, demand skewness, and user scales. SAC is omitted from these robustness checks due to its lack of competitiveness in the primary evaluation. To ensure that performance variations strictly reflect generalization capability rather than additional learning or relaxed executability, all sweeps adhere to a frozen trajectory protocol with identical steady-state initialization and strict action verification.


\subsubsection{Two-BS zero-shot generalization}
\label{subsubsec:two_bs_zero_shot}

We begin with a two-BS capacity sweep, since changing $C_b$ affects both prompt length and the admissible action space.
We evaluate transfer from the training capacity $C_b{=}10$ to $C_b\in\{10,15,20,25,30\}$ without retraining.
As shown in Table~\ref{tab:two_bs_zero_shot_cache}, GRPO generalizes robustly across the full range without retraining, despite the change in both cache-state dimensionality and feasible single-swap actions.
The advantage is most pronounced at small capacities, where storage pressure is highest and each replacement decision has a large opportunity cost: a myopic eviction can permanently displace a moderately popular content that would be repeatedly requested in the near future.
As $C_b$ increases, all methods improve and the gap to exhaustive search naturally narrows, because a larger cache reduces the need for frequent evictions and makes performance less sensitive to fine-grained replacement choices.

\begin{table}[htbp]
    \centering
    \caption{Summary of Parameter Sweep Ranges}
    \label{tab:parameters}
    \normalsize
    \renewcommand{\arraystretch}{1.14}
    \begin{tabularx}{\columnwidth}{@{} X l @{}} 
        \toprule
        \textbf{Parameter} & \textbf{Search Range / Values} \\
        \midrule
        Cache capacity ($C_b$)     & $\{10, 15, 20, 25, 30\}$ \\
        Library size ($F$)         & $\{100, 300, 500, 700, 900, 1100\}$ \\
        Zipf skew ($\alpha$)       & $\{0.6, 0.8, 1.0, 1.2, 1.4, 1.6\}$ \\
        Users ($U$)                & $\{20, 40, 60, 80, 100\}$ \\
        \bottomrule
    \end{tabularx}
\end{table}

\subsubsection{Five-BS robustness sweeps}
\label{subsubsec:five_bs_sweeps}

We next stress-test robustness in the more realistic five-BS cooperative setting.
Here, generalization is not only about prompt length, but also about whether the learned policy remains effective under
different levels of (i) storage pressure ($C_b$), (ii) content diversity ($F$),
(iii) demand concentration ($\alpha$), and (iv) request load and diversity ($U$).
Each of these modifies the environment dynamics and the admissible replacement set in a different way,
therefore jointly probing whether the policy has learned reusable coordination principles.

To evaluate the robustness of the proposed method, we perform a comprehensive parameter sweep without retraining the model, while strictly adhering to the frozen-trajectory protocol and ensuring strict executability. The specific ranges for the parameters $C_b$, $F$, $\alpha$, and $U$ are summarized in Table~\ref{tab:parameters}.

\paragraph{Cache capacity $C_b$}
Fig.~\ref{fig:five_bs_cache_capacity} varies $C_b$.
As $C_b$ increases, the average cooperative hit rate improves for all methods due to reduced eviction pressure.
Across the sweep, GRPO consistently outperforms SFT and the classical heuristics, and its margin over LFU/LRU is most visible at small-to-moderate capacities (e.g., $C_b{=}10$--$20$), where storage is scarce and redundant replication across overlapping BSs is most costly.
In the mid-capacity regime ($C_b{=}15$--$25$), GRPO is highly competitive with the exhaustive reference, indicating that the learned coordination can effectively diversify placements while still keeping necessary replication for frequently requested head contents.
When $C_b$ becomes large (e.g., $C_b{=}30$), all methods converge because most high-demand contents can be retained simultaneously, so the gap between GRPO and simple heuristics naturally shrinks.

\begin{figure}[t]
  \centering
  \safeincludegraphics[width=0.88\columnwidth]{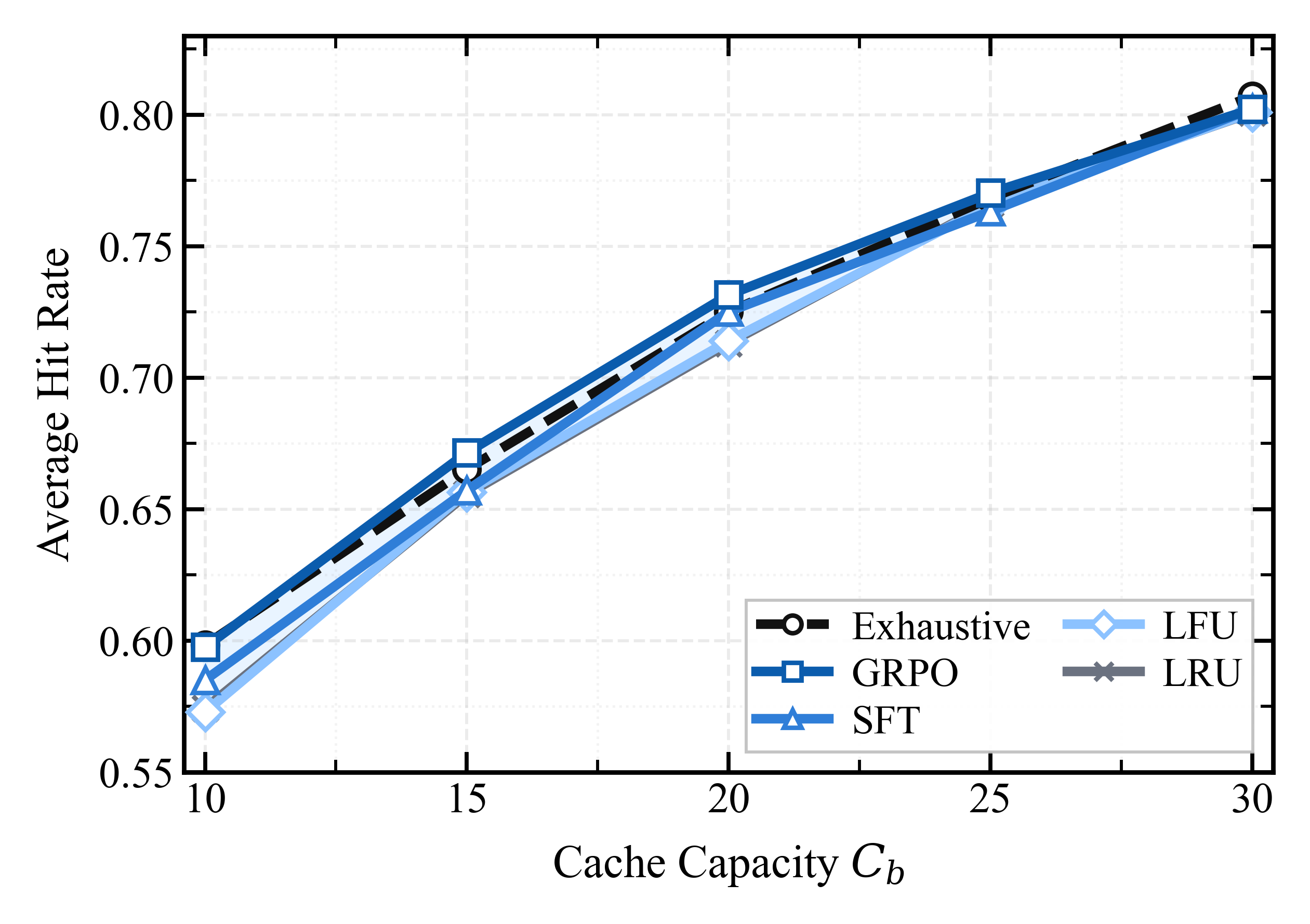}
  \caption{Five-BS robustness sweep: cooperative hit rate vs.\ cache capacity $C_b$ (zero-shot, no retraining).}
  \label{fig:five_bs_cache_capacity}
\end{figure}

\paragraph{Library size $F$}
Fig.~\ref{fig:five_bs_library_size} increases $F$.
This sweep probes robustness to content diversity and sparse reuse.
As $F$ grows, the cache-to-library ratio shrinks and repeated requests become less frequent, so all methods degrade substantially (from around $0.60$ at $F{=}100$ to around $0.44$--$0.45$ at $F{=}1100$).
In this regime, purely frequency/recency based rules can be brittle: LFU can accumulate stale counts and over-retain transiently popular contents, while LRU can overreact to one-off requests and induce churn.
GRPO degrades more gracefully and remains the strongest among the learned/heuristic baselines across the sweep, reflecting that the look-ahead training objective encourages replacements that improve future cooperative hits instead of reacting myopically to noisy single-slot evidence.
We also observe small crossings between GRPO and the exhaustive reference at the largest $F$, which are consistent with finite-horizon evaluation noise; overall, GRPO tracks the strongest reference closely while maintaining clear gains over SFT/LFU/LRU.

\begin{figure}[t]
  \centering
  \safeincludegraphics[width=0.88\columnwidth]{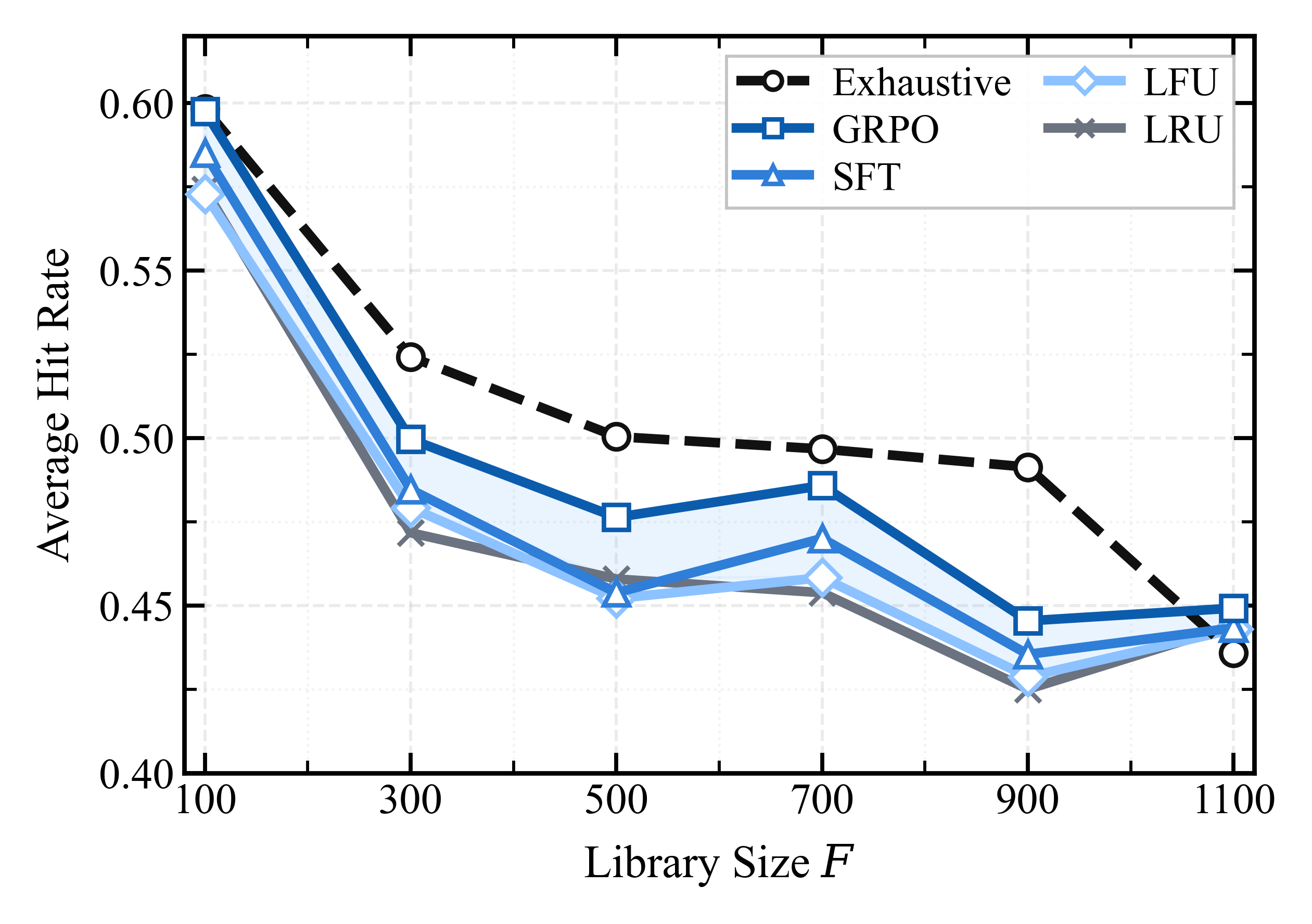}
  \caption{Five-BS robustness sweep: cooperative hit rate vs.\ library size $F$ (zero-shot, no retraining).}
  \label{fig:five_bs_library_size}
\end{figure}

\paragraph{Zipf skew $\alpha$}
Fig.~\ref{fig:five_bs_Zipf_sweep} varies $\alpha$, controlling demand concentration.
As $\alpha$ increases, demand becomes more concentrated on the head, so all methods improve substantially.
When the distribution is flatter-to-moderate (small $\alpha$, e.g., $0.6$--$1.0$), the long tail contributes significantly and cooperative coordination across overlapping BSs becomes essential; in this regime, GRPO maintains a consistent advantage over SFT and both heuristics by learning coverage-oriented placements under limited cache budgets.
As $\alpha$ becomes large, caching the head dominates performance and simple heuristics become stronger, so the advantage of sophisticated coordination diminishes and the curves move closer.
We further note that around $\alpha{=}1.4$ GRPO can match or slightly exceed the exhaustive reference in our traces, while for very concentrated demand (e.g., $\alpha{=}1.6$) the exhaustive reference remains highest and all methods are close due to the dominance of head contents.

\begin{figure}[t]
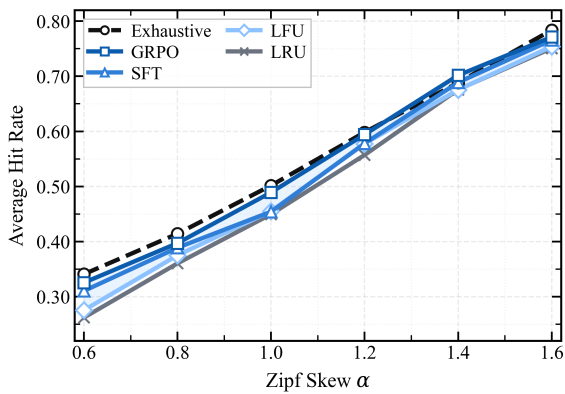

  \centering
  \safeincludegraphics[width=0.88\columnwidth]{five_bs_Zipf_sweep.png}
  \caption{Five-BS robustness sweep: cooperative hit rate vs.\ Zipf skew $\alpha$ (zero-shot, no retraining).}
  \label{fig:five_bs_Zipf_sweep}
\end{figure}

\paragraph{Number of Users $U$}
Fig.~\ref{fig:five_bs_user_sweep} varies the number of users $U$, which changes both the per-slot request volume and the diversity of the admissible insertion pools $\{\mathcal{R}^{(b)}_t\}$.
This introduces competing effects: a larger population provides richer statistical evidence to identify truly popular content, but simultaneously increases the instantaneous diversity of requests, heightening the risk of cache pollution.
We observe that the Single-step Exhaustive search performs robustly at lower user densities ($U \le 40$) but begins to degrade or stagnate as the population grows.
In contrast, GRPO exhibits superior scalability, overtaking the Exhaustive baseline at $U=60$ and achieving its peak performance at $U=80$.
This result highlights a critical advantage of the proposed method: in high-density regimes where $\mathcal{R}^{(b)}_t$ is large and noisy, the myopic Exhaustive search tends to "chase" transient requests, leading to suboptimal replacements. Conversely, GRPO's multi-step look-ahead training effectively filters this noise, prioritizing content with sustained cooperative utility over immediate but transient gains. Meanwhile, classical heuristics and SFT generally yield lower hit rates, confirming the necessity of reward-driven alignment for complex coordination.

\begin{figure}[t]
  \centering
  \safeincludegraphics[width=0.88\columnwidth]{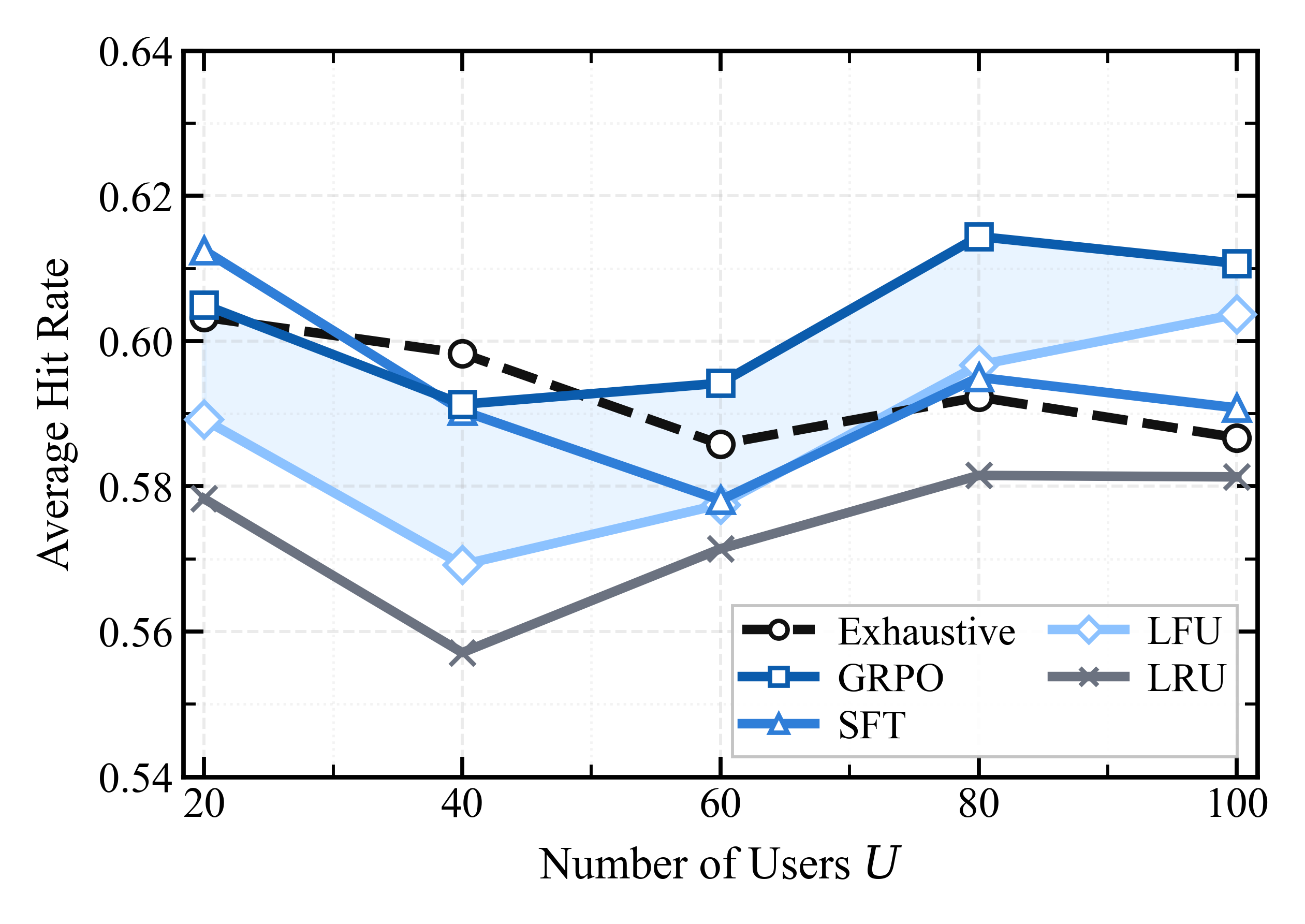}
  \caption{Five-BS robustness sweep: cooperative hit rate vs.\ number of users $U$ (zero-shot, no retraining).}
  \label{fig:five_bs_user_sweep}
\end{figure}

\paragraph{Summary}
Across both system scales and all zero-shot robustness sweeps, GRPO maintains robust improvements over SFT and classical eviction heuristics under strict executability.
Notably, the relative gains are most pronounced in the hardest regimes (small $C_b$, large $F$, and flatter popularity), where myopic local rules are most likely to waste capacity through redundant caching or cache pollution and where multi-step cooperative credit assignment is most valuable.
The two-BS capacity transfer confirms that the learned policy generalizes even when the prompt structure and action space change substantially,
while the five-BS sweeps demonstrate robustness to storage pressure, library diversity, demand skew, and user-driven request dynamics.

\section{Conclusion}
\label{sec:conclusion}

In this paper, we investigated cooperative multi-BS reactive cache replacement in wireless networks using an LLM-driven centralized coordinator. We formulated the task as a strictly constrained text-generation problem and proposed a two-stage alignment strategy combining SFT with GRPO to enforce feasibility and maximize long-term cooperative gains. We showed that the proposed framework effectively masters complex coordination logic without explicit modeling of the demand distribution. Numerical results under a rigorous frozen-trajectory protocol show that, in standard scenarios, the learned policy closely tracks the single-step exhaustive reference while consistently outperforming classical heuristics. Furthermore, we provided experimental evidence of the controller's zero-shot transferability, showing robust performance retention across varying cache capacities, library sizes, and user densities. Future work includes handling variable-sized content and incorporating inference latency constraints.

\bibliographystyle{IEEEtran}
\bibliography{references}

\appendices

\section{Proof of Theorem \ref{thm:shaping}}
\label{app:proof_thm1}

This appendix provides the detailed proof for the two properties claimed in Theorem~\ref{thm:shaping}. We demonstrate that maximizing the simplified shaped reward is mathematically equivalent to maximizing the potential gain (Decision Invariance) and that the witness-based penalty mechanism preserves the preference order among write actions while suppressing suboptimal passivity.

\subsection{Property 1: Decision Invariance}
Fix a decision slot $t$ and the frozen request window $\mathcal{Q}^{(t+1:t+H)}$ utilized by the look-ahead evaluator. Under the frozen-trajectory protocol, $\mathcal{Q}^{(t+1:t+H)}$ is exogenous and independent of the action chosen at slot $t$.

For any executable action $a_t$, let $\mathbf{X}^{(t+1)}(a_t)$ denote the cache configuration resulting from applying $a_t$ to the current cache state $\mathbf{X}^{(t)}$. Define the augmented current state as $\bar{s}^{(t)}\triangleq (s_t,\mathcal{Q}^{(t+1:t+H)})$ and the post-action augmented state as $\bar{s}^{(t+1)}(a_t)\triangleq (s_{t+1}(a_t),\mathcal{Q}^{(t+1:t+H)})$. Note that the request window remains invariant because Eq.~\eqref{eq:delta_perf_def} evaluates both cache configurations against the same future sequence.

Given the potential function $\Phi(\bar{s}^{(t)}) \triangleq \mathcal{H}_H(\mathbf{X}^{(t)}; \mathcal{Q}^{(t+1:t+H)})$, we can rewrite the performance gain in Eq.~\eqref{eq:delta_perf_def} for any $a_t$ as:
\begin{equation}
\Delta_{\text{perf}}(a_t) = \Phi\big(\bar{s}^{(t+1)}(a_t)\big) - \Phi(\bar{s}^{(t)}).
\end{equation}
Since $\Phi(\bar{s}^{(t)})$ depends exclusively on the pre-decision cache state and the frozen request window, it is constant with respect to the optimization variable $a_t$. Consequently, the optimization objective simplifies as follows:
\begingroup
\setlength{\abovedisplayskip}{3pt}
\setlength{\belowdisplayskip}{3pt}
\begin{equation}
\begin{aligned}
    & \argmax_{a_t} \Delta_{\text{perf}}(a_t) \\
    = & \argmax_{a_t} \Big[ \Phi\big(\bar{s}^{(t+1)}(a_t)\big) - \Phi(\bar{s}^{(t)}) \Big] \\
    = & \argmax_{a_t} \Phi\big(\bar{s}^{(t+1)}(a_t)\big).
\end{aligned}
\end{equation}
\endgroup
This derivation confirms the decision invariance property: maximizing the shaped reward is equivalent to maximizing the potential of the next state.

\subsection{Property 2: Order Preservation and Witness-Based \NoOp\ Demotion}
We analyze the preservation of reward ranking through two distinct cases.

\subsubsection*{Case 1: Ranking between two write actions}
Consider any two feasible write actions $w_1$ and $w_2$. By definition, the penalty term satisfies $P_t(w_1) = P_t(w_2) = 0$. The difference in their shaped rewards is:
\begin{equation}
\tilde{R}_t(w_1) - \tilde{R}_t(w_2) = \Delta_{\text{perf}}(w_1) - \Delta_{\text{perf}}(w_2).
\end{equation}
Thus, $\tilde{R}_t(w_1) > \tilde{R}_t(w_2)$ if and only if $\Delta_{\text{perf}}(w_1) > \Delta_{\text{perf}}(w_2)$, which implies that the preference ordering among write actions is strictly preserved.

\subsubsection*{Case 2: Witness-based demotion of \NoOp}
By definition, a \NoOp\ action implies $\mathbf{X}^{(t+1)} = \mathbf{X}^{(t)}$, yielding $\Delta_{\text{perf}}(\NoOp) = 0$.
Assume the training-time witness indicates that a beneficial write exists, i.e., $a_t^* \neq \NoOp$, and there exists a write action $w$ such that $\Delta_{\text{perf}}(w) > \Delta_{\text{perf}}(\NoOp) = 0$.
\begin{itemize}
    \item For the write action $w$:
    \begin{equation}
    \tilde{R}_t(w) = \Delta_{\text{perf}}(w) + P_t(w) = \Delta_{\text{perf}}(w) > 0.
    \end{equation}
    \item For the \NoOp\ action: Because the witness condition $a_t^* \neq \NoOp$ is active, the opportunity penalty is triggered, i.e., $P_t(\NoOp) = \lambda_{\text{opp}} < 0$. Thus:
    \begin{equation}
    \tilde{R}_t(\NoOp) = 0 + P_t(\NoOp) < 0.
    \end{equation}
\end{itemize}
Combining these results yields the strict inequality:
\begin{equation}
\tilde{R}_t(w) > 0 > \tilde{R}_t(\NoOp).
\end{equation}
This confirms that \NoOp\ is strictly demoted relative to any write action that yields a positive performance gain. 
\hfill \IEEEQEDclosed 

\section{Scalability of Joint Exhaustive Search}
\label{app:joint_exhaustive}

\begin{proposition}[Exponential growth of the joint action space]
\label{prop:joint_action_growth}
At slot $t$, let $\mathcal{A}_{b,t}$ denote the feasible action set of BS $b$ under the single-swap interface, and let $\mathcal{A}_t$ denote the joint action space of all $B$ BSs.
The joint action space is defined as $\mathcal{A}_t=\prod_{b\in\mathcal{B}}\mathcal{A}_{b,t}$, with cardinality $|\mathcal{A}_t|=\prod_{b\in\mathcal{B}}|\mathcal{A}_{b,t}|$.
If each BS admits at least one feasible write action in addition to \NoOp\ (i.e., $|\mathcal{A}_{b,t}|\ge 2$ for all $b$), then $|\mathcal{A}_t|\ge 2^{B}$, implying that the complexity of joint exhaustive enumeration grows exponentially with the number of BSs.
\end{proposition}

\begin{proof}
By definition, the controller outputs exactly one action per BS. Thus, a joint action is represented as the tuple:
\begin{equation}
    a_t = \left( a_t^{(1)}, \dots, a_t^{(B)} \right), \quad \text{where } a_t^{(b)} \in \mathcal{A}_{b,t}.
\end{equation}
Consequently, the joint feasible set is the Cartesian product of the individual feasible sets:
\begin{equation}
    \mathcal{A}_t = \mathcal{A}_{1,t} \times \cdots \times \mathcal{A}_{B,t} = \prod_{b \in \mathcal{B}} \mathcal{A}_{b,t}.
\end{equation}
The cardinality of the joint action space is the product of the cardinalities of the individual action sets:
\begin{equation}
    |\mathcal{A}_t| = \prod_{b\in\mathcal{B}} |\mathcal{A}_{b,t}|.
\end{equation}
Given the condition that $|\mathcal{A}_{b,t}| \ge 2$ for all $b \in \mathcal{B}$, we derive the lower bound:
\begin{equation}
    |\mathcal{A}_t| = \prod_{b\in\mathcal{B}} |\mathcal{A}_{b,t}| \ge \prod_{b\in\mathcal{B}} 2 = 2^{B}.
\end{equation}
This exponential lower bound renders joint exhaustive search computationally intractable for large $B$.
\end{proof}

\end{document}